\documentclass[a4paper,12pt]{article}

\usepackage{abstract}
\usepackage{amsfonts}
\usepackage{amsmath}
\usepackage{bm}
\usepackage{amssymb}
\usepackage{amsthm}
\usepackage[titletoc,toc,title]{appendix}
\usepackage{array}
\usepackage{authblk}
\usepackage{bbm}
\usepackage{caption}
\usepackage{dsfont}
\usepackage{setspace}
\usepackage{enumitem}
\usepackage[margin=1in]{geometry}
\usepackage{graphicx}
\usepackage{hyperref}
\usepackage{layout}
\usepackage{multirow}
\usepackage[round]{natbib}
\usepackage{leftidx}
\usepackage{subcaption}
\usepackage{tikz}
\usetikzlibrary{shapes,shadows,arrows,positioning,decorations.markings,decorations.pathreplacing}
\DeclareMathOperator*{\esssup}{ess\,sup}
\DeclareMathOperator*{\essinf}{ess\,inf}

\DeclareMathOperator*{\argmin}{arg\,min}

\makeatletter
\providecommand\phantomcaption{\caption@refstepcounter\@captype}
\makeatother

\makeatletter
\def\namedlabel#1#2{\begingroup
#2%
\def\@currentlabel{#2}%
\phantomsection\label{#1}\endgroup
}

\makeatother

\tikzset{
->-/.style={decoration={
markings,\theoremstyle{plain}
\newlist{casenv}{enumerate}{4}
\setlist[casenv]{leftmargin=*,align=left,widest={iiii}}
\setlist[casenv,1]{label={{\itshape\ \casename} \arabic*.},ref=\arabic*}
\setlist[casenv,2]{label={{\itshape\ \casename} \roman*.},ref=\roman*}
\setlist[casenv,3]{label={{\itshape\ \casename\ \alph*.}},ref=\alph*}
\setlist[casenv,4]{label={{\itshape\ \casename} \arabic*.},ref=\arabic*}
mark=at position .5 with {\arrow{>}}},postaction={decorate}},
-<-/.style={decoration={
markings,
mark=at position .5 with {\arrow{<}}},postaction={decorate}},
}

\title{\bf Strategic Risk Reduction:\\Self-Protection and Self-Insurance}
\author[$\star$]{Wing Fung Chong}
\affil[$\star$]{Department of Statistics and Actuarial Science, School of Computing and Data Science, The University of Hong Kong, Hong Kong, China. Email: chongwf@hku.hk.}

\begin{document}

\sloppy

\theoremstyle{definition}
\newtheorem{theorem}{Theorem}[section]
\newtheorem{corollary}[theorem]{Corollary}
\newtheorem{lemma}[theorem]{Lemma}
\newtheorem{proposition}[theorem]{Proposition}

\newtheorem{definition}{Definition}[section]
\newtheorem{problem}{Problem}[section]
\newtheorem{remark}{Remark}[section]
\newtheorem{example}{Example}[section]
\setcounter{section}{0}

\maketitle

\begin{abstract}
This paper studies how a risk holder should combine self-protection and self-insurance strategies when market insurance is absent. Self-protection reduces loss frequency, while self-insurance reduces loss severity. The risk holder incurs a joint risk-reduction cost that allows technological interaction between the two strategies and evaluates residual risk using either Value-at-Risk or Tail Value-at-Risk. In a Bernoulli model, we show that Value-at-Risk leads to a threshold-driven solution in which the optimal strategy is either no risk reduction, pure self-protection, or pure self-insurance, thereby exhibiting a substitution-type structure between the two risk-reduction strategies. By contrast, although Tail Value-at-Risk also admits a left-region/right-region decomposition, its left-region problem creates a direct residual frequency-severity interaction, making the local problem non-convex even in the Bernoulli setting. We solve this problem using an isoquant geometry method based on the marginal-balance curves for self-protection and self-insurance. The analysis identifies boundary, extreme constrained, touching, and crossing candidates, and shows how the confidence level and the cost technology determine whether self-protection and self-insurance behave as substitutes or complements. Illustrative examples compare the Value-at-Risk and Tail Value-at-Risk strategies, show how the confidence level changes the relevant isoquant geometry, and demonstrate that multiple crossings may generate non-unique optimal joint risk-reduction strategies.


%
\end{abstract}
\section{Introduction}

Recent emerging and large-scale risks, including extreme weather events, pandemics, cyber threats, and operational disruptions, have renewed interest in how individuals, firms, and institutions can reduce their own exposure before losses occur. In many such situations, market insurance may be unavailable, incomplete, expensive, or subject to exclusions. Two classical forms of ex-ante risk reduction are self-protection and self-insurance. Self-protection reduces the probability of a loss, while self-insurance reduces the severity of a loss conditional on its occurrence. Using cyber risk as an example, firewalls, access controls, phishing training, and endpoint monitoring are self-protection strategies, while backups, disaster recovery, incident response, and business-continuity planning are self-insurance strategies. Using flood risk as another example, flood barriers, drainage improvements, sealed doors, and backflow valves are self-protection strategies, while elevating electrical equipment, using waterproof materials, relocating valuable assets, and preparing post-flood restoration capacity are self-insurance strategies. This paper studies how a risk holder should strategically combine these two risk-reduction activities when market insurance is set aside.

The economic analysis of self-protection and self-insurance originates with the seminal work of \cite{EhrlichBecker1972}. They introduce a state-preference framework to study the interaction among self-protection, self-insurance, and market insurance. A central insight is that self-insurance and market insurance tend to be substitutes, whereas self-protection and market insurance may be complements or substitutes depending on the underlying environment. Subsequent work has developed this line of research in several directions. \cite{DionneEeckhoudt1985} study how increased risk aversion affects self-protection and self-insurance, showing that the effect is more robust for self-insurance than for self-protection. \cite{Lee1998} considers self-insurance-cum-protection, recognizing that many real-world activities may reduce both loss frequency and loss severity simultaneously. \cite{MuermannKunreuther2008} study optimal investment in self-protection by insured individuals when they face interdependencies through possible contamination from others. More recently, \cite{HofmannPeter2016} revisit self-protection and self-insurance in a two-period setting with savings and consumption smoothing, showing that the timing of investment and the presence of endogenous savings affect the comparative statics of risk-reduction effort. \cite{ChiPeterQi2026} further disentangle risk and time in the analysis of optimal prevention, emphasizing how dynamic considerations shape prevention incentives.

A broader and more recent literature has examined self-protection, self-insurance, prevention, and insurance demand under alternative preference models, risk measures, and dynamic risk environments. \cite{BensalemHernandezSantibanezKaziTani2020} study prevention effort and insurance demand under coherent risk measures in a Stackelberg setting with price incentives, giving special attention to self-protection and self-insurance. \cite{GauchonLoiselRulliereTrufin2020} and \cite{GauchonLoiselRulliereTrufin2021} analyze optimal prevention in risk models, including applications to ruin minimization and large-risk settings. \cite{BensalemHernandezSantibanezKaziTani2023} develop a continuous-time model of self-protection in which prevention effort dynamically reduces risk exposure. Surveys by \cite{CourbagePeterReyTreich2025} and \cite{Peter2025} provide comprehensive discussions of prevention, self-protection, self-insurance, and related insurance-economic models. A number of recent papers further study self-protection, self-insurance, prevention, insurance demand, and related risk-management problems, including \cite{BoonenNgNguyen2025}, \cite{BoonenZheng2025b}, \cite{ChenZhangZhu2025SSRN}, \cite{ChongFengHuZhang2025CyberRisk}, \cite{BoonenZheng2025a}, \cite{ChenZhangZhu2025EJOR}, and \cite{LiWangZhang2025}.

Most of the existing literature studies self-protection or self-insurance through marginal comparative statics, often in the presence of market insurance, or within self-insurance-cum-protection models in which a single prevention activity reduces both loss frequency and loss severity simultaneously. Such models mechanically tie the two margins together. By contrast, this paper focuses on the risk holder's direct strategic choice of residual risk when market insurance is absent. The risk holder chooses both a residual loss frequency and a residual loss severity. Equivalently, the risk holder chooses the amount of self-protection and the amount of self-insurance. This formulation allows us to ask a simple but fundamental question: if a risk holder must manage risk internally, how should self-protection and self-insurance be combined?

We consider a Bernoulli model. Before risk reduction, the loss is either zero or a positive amount. The risk holder chooses a residual probability of the positive loss and a residual severity conditional on the positive loss. Self-protection corresponds to lowering the residual probability, and self-insurance corresponds to lowering the residual severity. The joint cost of risk reduction is specified as a function of the amount of frequency reduction and the amount of severity reduction. This formulation permits technological interaction between self-protection and self-insurance. When increasing one type of risk reduction raises the marginal cost of the other, the two activities are technological substitutes in the cost technology. When increasing one type of risk reduction lowers the marginal cost of the other, they are technological complements in the cost technology. When increasing one type of risk reduction does not change the marginal cost of the other, the two activities are technologically independent in the cost technology.

Returning to the cyber-risk example, when prevention and recovery compete for the same IT budget, staff, and implementation capacity, they are substitutes in the cost technology. When a cybersecurity upgrade creates shared information and fixed setup costs, making prevention and recovery cheaper to implement together, they are complements in the cost technology. When prevention and recovery are provided by separate vendors with independent pricing and implementation, they are independent in the cost technology. Similarly, in the flood-risk example, flood-prevention works and severity-reduction renovations may be substitutes if they compete for the same renovation budget and contractors; they may be complements if a building retrofit creates shared design, engineering, and construction costs; and they may be independent if prevention devices and severity-reduction measures are purchased and installed through separate channels.

The paper studies two risk-measure objectives, namely, Value-at-Risk and Tail Value-at-Risk. Although the loss model is Bernoulli, the resulting optimization problems are quite different. Under Value-at-Risk, the problem is governed by a probability threshold. If the residual loss probability lies below this threshold, the positive loss is outside the Value-at-Risk tail. If the residual loss probability exceeds the threshold, the residual severity enters the objective. This threshold structure leads to a simple characterization. The optimal Value-at-Risk strategy is either no risk reduction, a pure self-protection threshold strategy, or a pure self-insurance strategy. The sign of the technological interaction between self-protection and self-insurance in the cost technology does not change this qualitative threshold-driven form. Thus, Value-at-Risk exhibits a substitution-type structure between the two risk-reduction strategies: the optimizer does not require their simultaneous use.

The Tail Value-at-Risk problem is fundamentally different. In the region where the residual loss probability lies below the Tail Value-at-Risk threshold, the objective contains a direct joint effect of residual frequency and residual severity. This tail-risk interaction creates a nonseparable optimization problem. As a result, the Tail Value-at-Risk objective need not be convex even in the Bernoulli setting. Standard convex analysis is therefore insufficient. The key economic object is the net interaction between residual frequency and residual severity. This net interaction combines the technological interaction in the risk-reduction cost with the tail-risk interaction generated by Tail Value-at-Risk.

The main technical contribution of the paper is to develop an isoquant geometry method for solving the Tail Value-at-Risk problem.\footnote{A related isoquant method was used by \cite{MuermannKunreuther2008} to study first- and second-best self-protection strategies between two interdependent risk holders facing possible contamination from each other.} The first marginal-balance isoquant consists of strategies at which the marginal effect of changing residual frequency is zero. The second marginal-balance isoquant consists of strategies at which the marginal effect of changing residual severity is zero. These two isoquants are the marginal-balance curves for self-protection and self-insurance. When the net interaction is positive, both isoquants are downward-sloping and generate a lower-left/upper-right separation geometry. When the net interaction is negative, both isoquants are upward-sloping, generating an upper-left/lower-right separation geometry, and the problem can be reduced to the positive-net-interaction case by reflecting the severity coordinate. When the net interaction is zero, the isoquants become coordinatewise: one is vertical and the other is horizontal.

Under positive net interaction, the paper gives a complete characterization of the left-region Tail Value-at-Risk optimizer. If at least one marginal-balance isoquant does not cut through the feasible rectangle, the problem reduces to a one-dimensional boundary minimization. If both isoquants cut through but do not intersect, the solution is one of two extreme constrained marginal-balance candidates. If the isoquants touch without crossing, an endpoint candidate may expand into a whole common marginal-balance component. If the isoquants cross, the relevant candidate set alternates across crossing components according to the initial vertical ordering of the isoquants and the parity of the number of crossings. These cases show that Tail Value-at-Risk can generate genuinely joint self-protection/self-insurance optima, unlike the Value-at-Risk problem, which remains fully threshold-driven.

The economic interpretation of the Tail Value-at-Risk solution depends on the interaction between the cost technology and the tail-risk term. If the cost function is supermodular, the cost-side interaction reinforces the Tail Value-at-Risk interaction. The net interaction is positive, and the marginal-balance curves are downward-sloping. In reduction variables, this corresponds to a substitution-type marginal tradeoff between self-protection and self-insurance, in which more self-protection is locally balanced against less self-insurance, and vice versa. If the cost function is sufficiently submodular, the cost-side complementarity may dominate the Tail Value-at-Risk interaction at lower confidence levels, leading to negative net interaction and upward-sloping marginal-balance curves. In reduction variables, this corresponds to a complementarity-type marginal tradeoff. As the confidence level increases, the Tail Value-at-Risk interaction becomes stronger and shifts the problem toward the positive-net-interaction regime.

The paper contributes to the literature in four ways. First, it isolates the internal risk-reduction problem of a risk holder who combines self-protection and self-insurance without relying on market insurance. This perspective is particularly relevant when insurance is unavailable, incomplete, or unaffordable. Second, it shows that Value-at-Risk and Tail Value-at-Risk lead to qualitatively different risk-reduction strategies. Value-at-Risk is fully threshold-driven, whereas Tail Value-at-Risk is governed by joint marginal-balance geometry. Third, it introduces an isoquant-based method for solving a non-convex Tail Value-at-Risk problem generated by the tail-risk interaction between residual frequency and residual severity. This method identifies boundary candidates, separated-isoquant candidates, touching components, and crossing components, and it gives a transparent economic interpretation of how cost-side and tail-risk interactions determine whether self-protection and self-insurance behave as substitutes or complements in the optimal risk-reduction strategy. Fourth, it provides illustrative examples that compare the Value-at-Risk and Tail Value-at-Risk strategies under a common quadratic cost specification, show how the confidence level shifts the Tail Value-at-Risk problem across negative, zero, and positive net-interaction regimes, and demonstrate that higher-order cost functions can generate multiple crossings and non-unique optimal joint strategies.

The rest of the paper is organized as follows. Section \ref{sec:problem} introduces the risk-reduction model, the cost technology, and the risk-measure formulation. Section \ref{sec:VaR} solves the Value-at-Risk problem and shows that its solution is governed by the probability threshold. Section \ref{sec:tvar} solves the Tail Value-at-Risk problem. It develops the isoquant geometry, analyzes positive, negative, and zero net interaction, and studies the effect of the confidence level on the net interaction. Section \ref{sec:illustrative} provides illustrative examples, while Section \ref{sec:conclusion} concludes and discusses possible extensions.

\section{Risk Reduction Problem}\label{sec:problem}
Let $t=0$ denote the current time, and let $\left(\Omega,\mathcal{F},\mathbb{P}\right)$ be a probability space. A risk holder (RH) faces a non-negative $\mathcal{F}$-measurable loss random variable $X$, whose realization is not observed until a future time $t=T>0$. Throughout this paper, the pre-mitigation risk $X$ is assumed to be non-degenerate Bernoulli under the objective probability measure $\mathbb{P}$. That is, there exist $\overline{p}\in\left(0,1\right)$ and $\overline{l}>0$ such that $\mathbb{P}\left(X=0\right)=1-\overline{p}$ and $\mathbb{P}\left(X=\overline{l}\right)=\overline{p}$. The value $\overline{p}$ represents the baseline probability of a positive loss, while $\overline{l}$ represents the baseline loss severity conditional on a positive loss.

At time $t=0$, the RH chooses ex-ante risk reduction strategies before the realization of $X$. Self-protection reduces the frequency of the loss and is represented by the choice of a residual loss probability $p\in\left[\underline{p},\overline{p}\right]$, where $\underline{p}\in\left(0,\overline{p}\right)$ is the lowest attainable positive loss probability. Self-insurance reduces the severity of the loss and is represented by the choice of a residual loss severity $l\in\left[\underline{l},\overline{l}\right]$, where $\underline{l}\in\left(0,\overline{l}\right)$ is the lowest attainable positive loss severity. Given an admissible strategy $\left(p,l\right)\in\left[\underline{p},\overline{p}\right]\times\left[\underline{l},\overline{l}\right]$, the mitigated risk $X\left(p,l\right)$ is Bernoulli with $\mathbb{P}\left(X\left(p,l\right)=0\right)=1-p$ and $\mathbb{P}\left(X\left(p,l\right)=l\right)=p$.

The pair $\left(p,l\right)$ is the residual-risk representation of the admissible strategy. For ease of exposition, it is also useful to introduce the corresponding reduction-amount representation, $x=\overline{p}-p$ and $y=\overline{l}-l$. Here, $x$ measures the amount of frequency reduction achieved through self-protection, while $y$ measures the amount of severity reduction achieved through self-insurance. The two representations are equivalent; the mapping $\left(p,l\right)\leftrightarrow\left(x,y\right)$ is one-to-one between $\left[\underline{p},\overline{p}\right]\times\left[\underline{l},\overline{l}\right]$ and $\left[0,\overline{x}\right]\times\left[0,\overline{y}\right]$, where $\overline{x}=\overline{p}-\underline{p}$ and $\overline{y}=\overline{l}-\underline{l}$. Throughout this paper, the analysis is mainly conducted in the residual-risk variables $\left(p,l\right)$, while the reduction variables $\left(x,y\right)$ are used when they provide clearer exposition.

For each admissible risk-reduction strategy $\left(p,l\right)$, the RH incurs an ex-ante risk-reduction cost. Since $p$ and $l$ are residual-risk variables, it is economically convenient to define this cost in terms of the reduction amounts $x$ and $y$. Let $c:\left[0,\overline{x}\right]\times\left[0,\overline{y}\right]\rightarrow\left[0,\infty\right)$ denote the joint cost function of self-protection and self-insurance, where $c\left(x,y\right)$ is the cost of reducing the positive-loss probability by $x$ and the positive-loss severity by $y$. We assume that $c\left(0,0\right)=0$, so that zero risk reduction incurs zero cost, and we assume that $c$ is continuous on its compact domain. Equivalently, in the residual-risk representation, for $\left(p,l\right)\in\left[\underline{p},\overline{p}\right]\times\left[\underline{l},\overline{l}\right]$, define $\pi\left(p,l\right)=c\left(\overline{p}-p,\overline{l}-l\right)$. Then $\pi\left(\overline{p},\overline{l}\right)=0$, and $\pi$ is continuous on its compact domain.

For mathematical simplicity, assume that $c$, and hence $\pi$, is twice continuously differentiable on the interior of its domain and admits the corresponding one-sided first- and second-order derivatives on the boundary. We denote these partial derivatives by $c_x,c_y,c_{xx},c_{xy},c_{yy}$, and $\pi_p,\pi_l,\pi_{pp},\pi_{pl},\pi_{ll}$. We assume that $c_x>0$, $c_y>0$, $c_{xx}>0$, and $c_{yy}>0$ on the interior of the domain, with the corresponding one-sided inequalities on the boundary. The first two conditions capture positive marginal costs of risk reduction, in either frequency or severity, when the other risk-reduction activity is held fixed. The last two conditions capture strictly increasing own marginal costs of risk reduction. Equivalently, $\pi_p=-c_x<0$, $\pi_l=-c_y<0$, $\pi_{pp}=c_{xx}>0$, and $\pi_{ll}=c_{yy}>0$.

The cross-partial derivative captures the technological relation between self-protection and self-insurance. A strictly supermodular joint cost function, meaning $\pi_{pl}=c_{xy}>0$, implies that increasing one type of risk reduction raises the marginal cost of the other. In this sense, self-protection and self-insurance are technological substitutes in the cost technology. Conversely, a strictly submodular joint cost function, meaning $\pi_{pl}=c_{xy}<0$, implies that increasing one type of risk reduction lowers the marginal cost of the other, so that self-protection and self-insurance are technological complements. When $\pi_{pl}=c_{xy}=0$ identically, the marginal cost of each risk-reduction activity is independent of the level of the other activity. This case is equivalent to the separable cost specification $c\left(x,y\right)=c_{\text{F}}\left(x\right)+c_{\text{S}}\left(y\right)$, up to the normalization by $c_{\text{F}}\left(0\right)=c_{\text{S}}\left(0\right)=0$. Equivalently, $\pi\left(p,l\right)=\pi_{\text{F}}\left(p\right)+\pi_{\text{S}}\left(l\right)$, where $\pi_{\text{F}}\left(p\right)=c_{\text{F}}\left(\overline{p}-p\right)$ and $\pi_{\text{S}}\left(l\right)=c_{\text{S}}\left(\overline{l}-l\right)$. In this case, $c_{\text{F}}$ and $c_{\text{S}}$ are continuous, strictly increasing, strictly convex, and twice continuously differentiable in the reduction amounts, while $\pi_{\text{F}}$ and $\pi_{\text{S}}$ are continuous, strictly decreasing, strictly convex, and twice continuously differentiable in the residual-risk variables.

\begin{example}\label{eg:cost_function}
Consider the quadratic family of joint cost functions: for $\left(x,y\right)\in\left[0,\overline{x}\right]\times\left[0,\overline{y}\right]$,
\begin{equation}
c\left(x,y\right)=\frac{1}{2}\left(Ax^2+2\delta xy+By^2\right)+ax+by,
\label{eq:cost_function}
\end{equation}
where $A>0$, $B>0$, $a>\max\left\{0,-\delta\overline{y}\right\}$, and $b>\max\left\{0,-\delta\overline{x}\right\}$. Then $c\left(0,0\right)=0$, and $c$ is continuous on its compact domain. Moreover, $c_x\left(x,y\right)=Ax+\delta y+a$, and $c_y\left(x,y\right)=\delta x+By+b$. The parameter restrictions on $a$ and $b$ ensure that $c_x>0$ and $c_y>0$ on the domain. Also, $c_{xx}=A>0$ and $c_{yy}=B>0$. Therefore, $c$ satisfies the required assumptions. Finally, since $c_{xy}=\delta$, we have that $c$ is strictly supermodular if $\delta>0$, is strictly submodular if $\delta<0$, and is separable if $\delta=0$.
\qed
\end{example}

Once an admissible risk-reduction strategy $\left(p,l\right)$ is chosen at time $t=0$, it is held until the future time $t=T$. Hence, viewed from time $t=0$, the RH's terminal wealth at time $t=T$ is the $\mathcal{F}$-measurable random variable
\begin{equation*}
W\left(p,l\right)=w-\pi\left(p,l\right)-X\left(p,l\right),
\end{equation*}
where $w\geq 0$ is the RH's initial wealth at time $t=0$. The status quo corresponds to the admissible pair $\left(\overline{p},\overline{l}\right)$, or equivalently $\left(x,y\right)=\left(0,0\right)$. Since $\pi\left(\overline{p},\overline{l}\right)=0$ and $X\left(\overline{p},\overline{l}\right)=X$, the status-quo wealth is $W\left(\overline{p},\overline{l}\right)=w-\pi\left(\overline{p},\overline{l}\right)-X\left(\overline{p},\overline{l}\right)=w-X$.

At time $t=0$, the RH chooses an admissible risk reduction strategy in order to maximize their preference over terminal wealth. Let ${\bm\theta}=(p,l)$ denote a generic admissible strategy, and write $W\left({\bm\theta}\right)=W\left(p,l\right)$. The RH's preference functional is denoted by $V\left(\cdot\right)$. Since the status quo corresponds to $\left(\overline{p},\overline{l}\right)$, a rational RH should only consider strategies that deliver a preference level at least as high as that of the status quo. Therefore, the RH solves
\begin{equation}
\sup_{{\bm\theta}\in\tilde{\Theta}}V\left(W\left({\bm\theta}\right)\right),
\label{eq:original_problem_rationality}
\end{equation}
where
\begin{equation}
\tilde{\Theta}=\left\{{\bm\theta}=\left(p,l\right)\in\left[\underline{p},\overline{p}\right]\times\left[\underline{l},\overline{l}\right]:V\left(W\left(\overline{p},\overline{l}\right)\right)\leq V\left(W\left({\bm\theta}\right)\right)\right\}.
\label{eq:tilde_Theta}
\end{equation}
For ${\bm\theta}\in\tilde{\Theta}$, the inequality $V\left(W\left(\overline{p},\overline{l}\right)\right)\leq V\left(W\left({\bm\theta}\right)\right)$ is referred to as the rationality constraint.

\subsection{General Preliminary Results}

\begin{proposition}\label{prop:non-empty_convex}
The status-quo strategy $\left(\overline{p},\overline{l}\right)$ belongs to $\tilde{\Theta}$. Hence, $\tilde{\Theta}$ is non-empty.
\end{proposition}
\begin{proof}
Since $\left(\overline{p},\overline{l}\right)\in\left[\underline{p},\overline{p}\right]\times\left[\underline{l},\overline{l}\right]$, and $V\left(W\left(\overline{p},\overline{l}\right)\right)\leq V\left(W\left(\overline{p},\overline{l}\right)\right)$, we have $\left(\overline{p},\overline{l}\right)\in\tilde{\Theta}$. Therefore, $\tilde{\Theta}\neq\emptyset$.
\end{proof}
\noindent
Consider the following unconstrained variant of Problem \eqref{eq:original_problem_rationality}:
\begin{equation}
\sup_{{\bm\theta}\in\Theta}V\left(W\left({\bm\theta}\right)\right),
\label{eq:original_problem}
\end{equation}
where
\begin{equation}
\Theta=\left\{{\bm\theta}=\left(p,l\right)\in\left[\underline{p},\overline{p}\right]\times\left[\underline{l},\overline{l}\right]\right\}=\left[\underline{p},\overline{p}\right]\times\left[\underline{l},\overline{l}\right].
\label{eq:original_feasible_set}
\end{equation}
Clearly, $\tilde{\Theta}\subseteq\Theta$.

\begin{corollary}\label{coro:non-empty_convex_theta}
The status-quo strategy $\left(\overline{p},\overline{l}\right)$ belongs to $\Theta$. Hence, $\Theta$ is non-empty.
\end{corollary}

Let $\tilde{\Theta}^*\subseteq\tilde{\Theta}$ denote the set of attainable optima of Problem \eqref{eq:original_problem_rationality}, and let $\Theta^*\subseteq\Theta$ denote the set of attainable optima of Problem \eqref{eq:original_problem}. Whenever the unconstrained problem has an attainable optimum, the rationality constraint does not affect the optimal solution; any unconstrained optimum must weakly dominate the status quo, which itself is feasible. The following proposition formalizes this observation.
\begin{proposition}\label{prop:remove_rationality}
Assume that $\Theta^*$ is non-empty. Then Problems \eqref{eq:original_problem_rationality} and \eqref{eq:original_problem} are equivalent in the sense that, for any ${\bm\theta}^*\in\Theta^*$,
\begin{equation}
\sup_{{\bm\theta}\in\tilde{\Theta}}V\left(W\left({\bm\theta}\right)\right)=\max_{{\bm\theta}\in\tilde{\Theta}}V\left(W\left({\bm\theta}\right)\right)=V\left(W\left({\bm\theta}^*\right)\right)=\max_{{\bm\theta}\in\Theta}V\left(W\left({\bm\theta}\right)\right)=\sup_{{\bm\theta}\in\Theta}V\left(W\left({\bm\theta}\right)\right),
\label{eq:all_equivalence}
\end{equation}
and $\tilde{\Theta}^*=\Theta^*$.
\end{proposition}
\begin{proof}
Let ${\bm\theta}^*\in\Theta^*$. For any ${\bm\theta}\in\tilde{\Theta}\subseteq\Theta$,
\begin{equation*}
V\left(W\left(\overline{p},\overline{l}\right)\right)\leq V\left(W\left({\bm\theta}\right)\right)\leq V\left(W\left({\bm\theta}^*\right)\right)=\max_{{\bm\theta}\in\Theta}V\left(W\left({\bm\theta}\right)\right)=\sup_{{\bm\theta}\in\Theta}V\left(W\left({\bm\theta}\right)\right),
\end{equation*}
which implies that ${\bm\theta}^*\in\tilde{\Theta}$. Hence, $\sup_{{\bm\theta}\in\tilde{\Theta}}V\left(W\left({\bm\theta}\right)\right)=\max_{{\bm\theta}\in\tilde{\Theta}}V\left(W\left({\bm\theta}\right)\right)=V\left(W\left({\bm\theta}^*\right)\right)$. This proves \eqref{eq:all_equivalence} and implies that ${\bm\theta}^*\in\tilde{\Theta}^*$. Therefore, $\Theta^*\subseteq\tilde{\Theta}^*$. Since $\Theta^*$ is non-empty, $\tilde{\Theta}^*$ is also non-empty.

Conversely, let $\tilde{\bm\theta}^*\in\tilde{\Theta}^*$. Since $\tilde{\Theta}\subseteq\Theta$, we have $\tilde{\bm\theta}^*\in\Theta$. By \eqref{eq:all_equivalence},
\begin{equation*}
V\left(W\left(\tilde{\bm\theta}^*\right)\right)=\max_{{\bm\theta}\in\tilde{\Theta}}V\left(W\left({\bm\theta}\right)\right)=\sup_{{\bm\theta}\in\tilde{\Theta}}V\left(W\left({\bm\theta}\right)\right)=\sup_{{\bm\theta}\in\Theta}V\left(W\left({\bm\theta}\right)\right)=\max_{{\bm\theta}\in\Theta}V\left(W\left({\bm\theta}\right)\right).
\end{equation*}
Thus $\tilde{\bm\theta}^*\in\Theta^*$, which proves that $\tilde{\Theta}^*\subseteq\Theta^*$. Hence, $\tilde{\Theta}^*=\Theta^*$.
\end{proof}
\noindent
By Proposition \ref{prop:remove_rationality}, whenever the unconstrained problem admits an attainable optimum, it is equivalent to solve Problem \eqref{eq:original_problem} over the rectangular feasible set $\Theta$ in \eqref{eq:original_feasible_set}. The following proposition provides a sufficient condition for attainability of the unconstrained problem.
\begin{proposition}\label{coro:upper_semicontinuous}
Assume that the map ${\bm\theta}\mapsto V\left(W\left({\bm\theta}\right)\right)$ is upper semicontinuous on $\Theta$. Then Problems \eqref{eq:original_problem_rationality} and \eqref{eq:original_problem} are equivalent, and $\tilde{\Theta}^*=\Theta^*$.
\end{proposition}
\begin{proof}
The feasible set $\Theta=\left[\underline{p},\overline{p}\right]\times\left[\underline{l},\overline{l}\right]$ is compact. Together with the upper semicontinuity assumption, the Weierstrass extreme value theorem implies that Problem \eqref{eq:original_problem} attains its maximum on $\Theta$. Therefore, $\Theta^*\neq\emptyset$. The result immediately follows from Proposition \ref{prop:remove_rationality}.
\end{proof}

\subsection{Preferences via Risk Measures}
Suppose that the RH evaluates terminal wealth through a risk measure $\rho$, in the sense that, for a generic terminal wealth random variable $Z$, $V\left(Z\right)=-\rho\left(-Z\right)$. Assume that $\rho$ is translation invariant. Hence, Problem \eqref{eq:original_problem_rationality} is equivalent to
\begin{equation}
\inf_{{\bm\theta}\in\tilde{\Theta}}\pi\left({\bm\theta}\right)+\rho\left(X\left({\bm\theta}\right)\right),
\label{eq:risk_problem}
\end{equation}
where the feasible set $\tilde{\Theta}$ is given in \eqref{eq:tilde_Theta}. Similarly, Problem \eqref{eq:original_problem} is equivalent to
\begin{equation}
\inf_{{\bm\theta}\in\Theta}\pi\left({\bm\theta}\right)+\rho\left(X\left({\bm\theta}\right)\right),
\label{eq:risk_problem_2}
\end{equation}
where the feasible set $\Theta$ is given in \eqref{eq:original_feasible_set}.

The following corollary gives a sufficient condition, in terms of the risk-measure objective, under which Proposition \ref{coro:upper_semicontinuous} applies.
\begin{corollary}\label{coro:lower_semicontinuous}
Assume that the map ${\bm\theta}\mapsto \rho\left(X\left({\bm\theta}\right)\right)$ is lower semicontinuous on $\Theta$. Then Problems \eqref{eq:risk_problem} and \eqref{eq:risk_problem_2} are equivalent, and $\tilde{\Theta}^*=\Theta^*$.
\end{corollary}
\begin{proof}
Since $\pi$ is continuous on $\Theta$, the map ${\bm\theta}\mapsto \pi\left({\bm\theta}\right)+\rho\left(X\left({\bm\theta}\right)\right)$ is lower semicontinuous on $\Theta$. Therefore, the map ${\bm\theta}\mapsto V\left(W\left({\bm\theta}\right)\right)$ is upper semicontinuous on $\Theta$. The result follows from Proposition \ref{coro:upper_semicontinuous}.
\end{proof}

This paper focuses on the cases in which $\rho$ is Value-at-Risk (VaR) or Tail Value-at-Risk (TVaR). For a generic random variable $Z$, let $F_Z\left(z\right)=\mathbb{P}\left(Z\leq z\right)$ and $S_Z\left(z\right)=\mathbb{P}\left(Z>z\right)=1-F_Z\left(z\right)$ denote its distribution function and survival function, respectively, for $z\in\left[\essinf Z,\esssup Z\right]$. For a confidence level $\alpha\in\left[0,1\right]$, VaR is defined by
\begin{align*}
\text{VaR}_{\alpha}\left(Z\right)=&\;
F_Z^{-1}\left(\alpha\right)=\inf\left\{z\in\left[\essinf Z,\esssup Z\right]:F_Z\left(z\right)\geq\alpha\right\}\\=&\;\inf\left\{z\in\left[\essinf Z,\esssup Z\right]:S_Z\left(z\right)\leq 1-\alpha\right\}=S_Z^{-1}\left(1-\alpha\right),
\end{align*}
with the convention that $\inf\emptyset=\esssup Z$. For $\alpha\in\left[0,1\right)$, TVaR is defined by
\begin{equation*}
\text{TVaR}_{\alpha}\left(Z\right)=\frac{1}{1-\alpha}\int_{\alpha}^{1}\text{VaR}_{\gamma}\left(Z\right)\text{d}\gamma,
\end{equation*}
and, for $\alpha=1$, define $\text{TVaR}_{1}\left(Z\right)=\esssup Z$.

In the case of VaR, for $\alpha\in\left[0,1\right]$ and ${\bm\theta}=\left(p,l\right)\in\Theta$,
\begin{equation*}
\text{VaR}_{\alpha}\left(X\left({\bm\theta}\right)\right)=\mathds{1}_{\left(1-\alpha,1\right]}\left(p\right)l=\begin{cases}
0,&\text{ for }p\in\left[0,1-\alpha\right],\\
l,&\text{ for }p\in\left(1-\alpha,1\right].
\end{cases}
\end{equation*}
This map is lower semicontinuous on $\Theta$. In the case of TVaR, for $\alpha\in\left[0,1\right)$ and ${\bm\theta}=\left(p,l\right)\in\Theta$,
\begin{equation*}
\text{TVaR}_{\alpha}\left(X\left({\bm\theta}\right)\right)=\min\left\{\frac{p}{1-\alpha},1\right\}l=
\begin{cases}
\dfrac{pl}{1-\alpha},&\text{ for }p\in\left[0,1-\alpha\right],\\
l,&\text{ for }p\in\left(1-\alpha,1\right].
\end{cases}
\end{equation*}
For $\alpha=1$, $\text{TVaR}_{1}\left(X\left({\bm\theta}\right)\right)=\esssup X\left({\bm\theta}\right)=l$. Thus ${\bm\theta}\mapsto\text{TVaR}_{\alpha}\left(X\left({\bm\theta}\right)\right)$ is continuous on $\Theta$ for all $\alpha\in\left[0,1\right]$.

By Corollary \ref{coro:lower_semicontinuous}, in the VaR case, it is sufficient to solve the unconstrained problem
\begin{equation}
\inf_{\left(p,l\right)\in\Theta}\pi\left(p,l\right)+\mathds{1}_{\left(1-\alpha,1\right]}\left(p\right)l.
\label{eq:VaR_problem}
\end{equation}
In the TVaR case with $\alpha\in\left[0,1\right)$, it is sufficient to solve the unconstrained problem
\begin{equation}
\inf_{\left(p,l\right)\in\Theta}\pi\left(p,l\right)+\min\left\{\frac{p}{1-\alpha},1\right\}l.
\label{eq:TVaR_problem}
\end{equation}
For $\alpha=1$, the corresponding TVaR problem is
\begin{equation}
\inf_{\left(p,l\right)\in\Theta}\pi\left(p,l\right)+l.
\label{eq:TVaR_problem_2}
\end{equation}

The VaR and TVaR formulations above lead to two distinct optimization problems. The VaR problem is governed by the threshold $p=1-\alpha$; the severity $l$ enters the objective only when the residual loss probability exceeds this threshold. Thus VaR generates a discontinuous, threshold-driven risk-reduction problem. By contrast, the TVaR problem is continuous for $\alpha\in\left[0,1\right]$, but, for $\alpha\in\left[0,1\right)$, it introduces a direct interaction between residual frequency and residual severity through the bilinear term $pl/\left(1-\alpha\right)$ whenever $p\in\left[0,1-\alpha\right]$. This term is the main source of the economic geometry of the TVaR problem and implies that the objective need not be convex. The following sections solve the VaR and TVaR problems by decomposing the feasible set according to the threshold $p=1-\alpha$, comparing the resulting local solutions, and interpreting the optimal strategies in terms of self-protection and self-insurance.

\section{Value-at-Risk}\label{sec:VaR}
This section solves the VaR problem \eqref{eq:VaR_problem} with the feasible set $\Theta$ in \eqref{eq:original_feasible_set}. For ${\bm\theta}=\left(p,l\right)\in\Theta$, denote the objective function in \eqref{eq:VaR_problem} by
\begin{equation*}
F\left({\bm\theta}\right)=F\left(p,l\right)=\pi\left(p,l\right)+\mathds{1}_{\left(1-\alpha,1\right]}\left(p\right)l.
\end{equation*}
The feasible set can be decomposed according to whether the residual loss probability lies below or above the VaR probability threshold $1-\alpha$. Define
\begin{equation}
\Theta^{\left(\text{L}\right)}=\left\{\left(p,l\right)\in\Theta:p\leq 1-\alpha\right\},
\label{eq:Theta^L}
\end{equation}
and
\begin{equation}
\Theta^{\left(\text{R}\right)}=\left\{\left(p,l\right)\in\Theta:1-\alpha<p\right\}.
\label{eq:Theta^R}
\end{equation}
On $\Theta^{\left(\text{L}\right)}$, the VaR term is zero, whereas on $\Theta^{\left(\text{R}\right)}$, the VaR term is equal to $l$. Thus, the corresponding local problems are
\begin{equation}
\inf_{\left(p,l\right)\in\Theta^{\left(\text{L}\right)}}F(p,l)=\inf_{\left(p,l\right)\in\Theta^{\left(\text{L}\right)}}\pi\left(p,l\right),
\label{eq:VaR_problem<=}
\end{equation}
and
\begin{equation}
\inf_{\left(p,l\right)\in\Theta^{\left(\text{R}\right)}}F(p,l)=\inf_{\left(p,l\right)\in\Theta^{\left(\text{R}\right)}}\pi\left(p,l\right)+l.
\label{eq:VaR_problem>}
\end{equation}

\begin{proposition}\label{lemma:VaR_problem<=}
Assume that $\Theta^{\left(\text{L}\right)}$ is non-empty. Then the unique optimal solution of Problem \eqref{eq:VaR_problem<=} is ${\bm\theta}^{*\left(\text{L}\right)}=\left(p^{*\left(\text{L}\right)},l^{*\left(\text{L}\right)}\right)=\left(\min\left\{\overline{p},1-\alpha\right\},\overline{l}\right)$.
\end{proposition}
\begin{proof}
Since $\Theta^{\left(\text{L}\right)}\neq\emptyset$, we have $\underline{p}\leq 1-\alpha$, and hence
\begin{equation*}
\Theta^{\left(\text{L}\right)}=\left[\underline{p},\min\left\{\overline{p},1-\alpha\right\}\right]\times\left[\underline{l},\overline{l}\right].
\end{equation*}
The cost function $\pi$ is strictly decreasing in both $p$ and $l$. Therefore, over $\Theta^{\left(\text{L}\right)}$, the value of $\pi\left(p,l\right)$ is minimized by choosing the largest feasible residual probability and the largest feasible residual severity. Hence, ${\bm\theta}^{*\left(\text{L}\right)}=\left(\min\left\{\overline{p},1-\alpha\right\},\overline{l}\right)$. Uniqueness follows from the strict monotonicity of $\pi$ in both arguments.
\end{proof}

\begin{proposition}\label{lemma:VaR_problem>}
Assume that $\Theta^{\left(\text{R}\right)}$ is non-empty. Define
\begin{equation}
\tilde{l}^*\in\argmin_{l\in\left[\underline{l},\overline{l}\right]}\pi\left(\overline{p},l\right)+l.
\label{eq:VaR_argmin}
\end{equation}
Then the unique optimal solution of Problem \eqref{eq:VaR_problem>} is ${\bm\theta}^{*\left(\text{R}\right)}=\left(p^{*\left(\text{R}\right)},l^{*\left(\text{R}\right)}\right)=\left(\overline{p},\tilde{l}^*\right)$. Moreover,
\begin{equation*}
\tilde{l}^*=
\begin{cases}
\underline{l},&\text{ if }0\leq \pi_l\left(\overline{p},\underline{l}\right)+1,\\
\tilde{l},&\text{ if }\pi_l\left(\overline{p},\underline{l}\right)+1<0<\pi_l\left(\overline{p},\overline{l}\right)+1,\\
\overline{l},&\text{ if }\pi_l\left(\overline{p},\overline{l}\right)+1\leq 0,
\end{cases}
\end{equation*}
where $\tilde{l}\in\left(\underline{l},\overline{l}\right)$ is the unique solution of $\pi_l\left(\overline{p},l\right)+1=0$ under the condition that $\pi_l\left(\overline{p},\underline{l}\right)+1<0<\pi_l\left(\overline{p},\overline{l}\right)+1$.
\end{proposition}
\begin{proof}
Since $\Theta^{\left(\text{R}\right)}\neq\emptyset$, we have $1-\alpha<\overline{p}$, and hence $\left(\overline{p},l\right)\in\Theta^{\left(\text{R}\right)}$ for every $l\in\left[\underline{l},\overline{l}\right]$. For any $\left(p,l\right)\in\Theta^{\left(\text{R}\right)}$ with $p<\overline{p}$, strict monotonicity of $\pi$ in $p$ gives $\pi\left(\overline{p},l\right)+l<\pi\left(p,l\right)+l$. Therefore, any optimizer of Problem \eqref{eq:VaR_problem>} must have $p^{*\left(\text{R}\right)}=\overline{p}$. The problem is thus reduced to
\begin{equation}
\inf_{l\in\left[\underline{l},\overline{l}\right]}\pi\left(\overline{p},l\right)+l.
\label{eq:var_one_dim_severity_problem}
\end{equation}
The objective is continuous on the compact interval $\left[\underline{l},\overline{l}\right]$, so a minimizer exists. Since $\pi_{ll}>0$, the one-dimensional function $l\mapsto \pi\left(\overline{p},l\right)+l$ is strictly convex on $\left[\underline{l},\overline{l}\right]$. Hence the minimizer $\tilde{l}^*$ is unique.

The stated first-order characterization follows from the strict convexity of the one-dimensional objective and the corresponding one-sided derivatives at the boundary. If the derivative at $\underline{l}$ is non-negative, the minimum is attained at $\underline{l}$. If the derivative at $\overline{l}$ is non-positive, the minimum is attained at $\overline{l}$. Otherwise, the unique minimizer is the interior point $\tilde{l}$ satisfying $\pi_l\left(\overline p,\tilde{l}\right)+1=0$.
\end{proof}

The next theorem gives the complete solution of the VaR problem.
\begin{theorem}\label{thm:var}
Let $\tilde{l}^*$ be defined as in Proposition \ref{lemma:VaR_problem>}. The optimal solution of the VaR problem \eqref{eq:VaR_problem} is characterized as follows.
\begin{enumerate}
\item[(i)] If $\alpha\in\left[0,1-\overline{p}\right]$, then ${\bm\theta}^*=\left(\overline{p},\overline{l}\right)$.
\item[(ii)] If $\alpha\in\left(1-\overline{p},1-\underline{p}\right]$, then ${\bm\theta}^*\in\argmin_{{\bm\theta}\in\left\{\left(1-\alpha,\overline{l}\right),\left(\overline{p},\tilde{l}^*\right)\right\}}F\left({\bm\theta}\right)$.
\item[(iii)] If $\alpha\in\left(1-\underline{p},1\right]$, then ${\bm\theta}^*=\left(\overline{p},\tilde{l}^*\right)$.
\end{enumerate}
\end{theorem}
\begin{proof}
Suppose first that $\alpha\in\left[0,1-\overline{p}\right]$. Then $\overline{p}\leq 1-\alpha$, so every feasible residual probability satisfies $p\leq 1-\alpha$. Hence $\Theta=\Theta^{\left(\text{L}\right)}$, and the problem reduces to Problem \eqref{eq:VaR_problem<=}. By Proposition \ref{lemma:VaR_problem<=}, the unique optimizer is ${\bm\theta}^*={\bm\theta}^{*\left(\text{L}\right)}=\left(\min\left\{\overline{p},1-\alpha\right\},\overline{l}\right)=\left(\overline{p},\overline{l}\right)$.

Next, suppose that $\alpha\in\left(1-\underline{p},1\right]$. Then $1-\alpha<\underline{p}$, so every feasible residual probability satisfies $1-\alpha<p$. Hence $\Theta=\Theta^{\left(\text{R}\right)}$, and the problem reduces to Problem \eqref{eq:VaR_problem>}. By Proposition \ref{lemma:VaR_problem>}, the unique optimizer is ${\bm\theta}^*={\bm\theta}^{*\left(\text{R}\right)}=\left(\overline{p},\tilde{l}^*\right)$.

Finally, suppose that $\alpha\in\left(1-\overline{p},1-\underline{p}\right]$. Then $\underline{p}\leq 1-\alpha<\overline{p}$, so both $\Theta^{\left(\text{L}\right)}$ and $\Theta^{\left(\text{R}\right)}$ are non-empty. By Proposition \ref{lemma:VaR_problem<=}, the optimal candidate on $\Theta^{\left(\text{L}\right)}$ is ${\bm\theta}^{*\left(\text{L}\right)}=\left(\min\left\{\overline{p},1-\alpha\right\},\overline{l}\right)=\left(1-\alpha,\overline{l}\right)$, with objective value $F\left({\bm\theta}^{*\left(\text{L}\right)}\right)=\pi\left(1-\alpha,\overline{l}\right)$. By Proposition \ref{lemma:VaR_problem>}, the optimal candidate on $\Theta^{\left(\text{R}\right)}$ is ${\bm\theta}^{*\left(\text{R}\right)}=\left(\overline{p},\tilde{l}^*\right)$, with objective value $F\left({\bm\theta}^{*\left(\text{R}\right)}\right)=\pi\left(\overline{p},\tilde{l}^*\right)+\tilde{l}^*$. Comparing these two objective values gives the stated result.
\end{proof}

Theorem \ref{thm:var} shows that VaR induces a threshold-driven structure. In the low-confidence region $\alpha\in\left[0,1-\overline{p}\right]$, the positive loss lies outside the VaR tail for every feasible strategy, so no risk reduction is optimal. In the high-confidence region $\alpha\in\left(1-\underline{p},1\right]$, the positive loss is always inside the VaR tail, and the optimal strategy involves no self-protection but may involve self-insurance. In the intermediate region $\alpha\in\left(1-\overline{p},1-\underline{p}\right]$, the RH compares a pure self-protection strategy that moves the residual probability to the VaR threshold, with a pure self-insurance strategy that keeps the residual probability at its baseline level. Thus, under VaR, optimal risk reduction is threshold-driven and does not require the simultaneous use of self-protection and self-insurance.

This qualitative structure is independent of the sign of the cross-partial derivative of the joint cost function. Whether $\pi_{pl}=c_{xy}>0$, $\pi_{pl}=c_{xy}<0$, or $\pi_{pl}=c_{xy}=0$, the VaR problem is still governed by the threshold $1-\alpha$. Recall that, economically, $\pi_{pl}=c_{xy}>0$ corresponds to self-protection and self-insurance being technological substitutes in the cost technology, because increasing one type of risk reduction raises the marginal cost of the other. Conversely, $\pi_{pl}=c_{xy}<0$ corresponds to technological complements, because increasing one type of risk reduction lowers the marginal cost of the other. The case $\pi_{pl}=c_{xy}=0$ corresponds to the separable specification, in which the marginal cost of each activity is independent of the level of the other. Under VaR, these technological relations do not change the threshold-driven form of the optimal solution; the RH either chooses no risk reduction, a pure self-protection threshold strategy, or a pure self-insurance strategy. In the right-region case, the residual severity choice $\tilde{l}^*$ is determined by the one-dimensional boundary problem \eqref{eq:var_one_dim_severity_problem}, rather than by the global sign of $\pi_{pl}=c_{xy}$. This feature contrasts with the TVaR problem studied next, where the bilinear tail-risk term $pl/\left(1-\alpha\right)$ creates a direct interaction between residual frequency and residual severity.

\section{Tail Value-at-Risk}\label{sec:tvar}

This section solves the TVaR problem \eqref{eq:TVaR_problem} with the feasible set $\Theta$ in \eqref{eq:original_feasible_set}. For ${\bm\theta}=\left(p,l\right)\in\Theta$, denote the objective function in \eqref{eq:TVaR_problem} by $G\left({\bm\theta}\right)=G\left(p,l\right)$. More explicitly, for $\alpha\in\left[0,1\right)$,
\begin{equation}
G\left(p,l\right)=\pi\left(p,l\right)+\min\left\{\frac{p}{1-\alpha},1\right\}l,
\label{eq:tvar_problem_alpha<1}
\end{equation}
whereas, for $\alpha=1$,
\begin{equation}
G\left(p,l\right)=\pi\left(p,l\right)+l.
\label{eq:tvar_problem_alpha=1}
\end{equation}

\subsection{Two Local Problems}

Let $\alpha\in\left[0,1\right)$. As in the VaR case, the feasible set is decomposed according to the threshold $1-\alpha$. Define $\Theta^{\left(\text{L}\right)}$ and $\Theta^{\left(\text{R}\right)}$ as in \eqref{eq:Theta^L} and \eqref{eq:Theta^R}; that is,
\begin{equation*}
\Theta^{\left(\text{L}\right)}=\left\{\left(p,l\right)\in\Theta:p\leq 1-\alpha\right\},
\end{equation*}
and
\begin{equation*}
\Theta^{\left(\text{R}\right)}=\left\{\left(p,l\right)\in\Theta:1-\alpha<p\right\}.
\end{equation*}
On the left region $\Theta^{\left(\text{L}\right)}$, the TVaR objective in \eqref{eq:tvar_problem_alpha<1} becomes
\begin{equation*}
G\left(p,l\right)=\pi\left(p,l\right)+\frac{pl}{1-\alpha}.
\end{equation*}
On the right region $\Theta^{\left(\text{R}\right)}$, it becomes
\begin{equation*}
G\left(p,l\right)=\pi\left(p,l\right)+l.
\end{equation*}
Therefore, for $\alpha\in\left[0,1\right)$, the TVaR problem can be solved by comparing the local solutions of
\begin{equation}
\inf_{\left(p,l\right)\in\Theta^{\left(\text{L}\right)}}G\left(p,l\right)=\inf_{\left(p,l\right)\in\Theta^{\left(\text{L}\right)}}\pi\left(p,l\right)+\frac{pl}{1-\alpha},
\label{eq:TVaR_problem<=}
\end{equation}
and
\begin{equation}
\inf_{\left(p,l\right)\in\Theta^{\left(\text{R}\right)}}G\left(p,l\right)=\inf_{\left(p,l\right)\in\Theta^{\left(\text{R}\right)}}\pi\left(p,l\right)+l.
\label{eq:TVaR_problem>}
\end{equation}
\noindent
The right-region problem \eqref{eq:TVaR_problem>} is identical to the right-region VaR problem in \eqref{eq:VaR_problem>}. Hence, by Proposition \ref{lemma:VaR_problem>}, whenever $\Theta^{\left(\text{R}\right)}$ is non-empty, its unique optimal solution is
${\bm\theta}^{*\left(\text{R}\right)}=\left(\overline{p},\tilde{l}^*\right)$, where $\tilde{l}^*$ is given in \eqref{eq:VaR_argmin}.

When $\alpha=1$, we have $1-\alpha=0<\underline{p}$, and hence every feasible residual loss probability satisfies $1-\alpha=0<p$. Therefore, $\Theta^{\left(\text{L}\right)}=\emptyset$ and $\Theta^{\left(\text{R}\right)}=\Theta$, and the TVaR problem is solved by the same right-region solution ${\bm\theta}^{*\left(\text{R}\right)}=\left(\overline{p},\tilde{l}^*\right)$.

Thus, the remainder of this section focuses on the case $\alpha\in\left[0,1\right)$. The main technical difficulty of the TVaR problem lies in solving the left-region problem \eqref{eq:TVaR_problem<=} when $\Theta^{\left(\text{L}\right)}\neq\emptyset$. The rest of this section studies this left-region problem.

\subsection{Isoquant Geometry}\label{sec:iso_geo}

This subsection develops the geometric structure of the left-region TVaR problem \eqref{eq:TVaR_problem<=} when $\alpha\in\left[0,1\right)$ and $\Theta^{\left(\text{L}\right)}\neq\emptyset$. The first-order derivatives of $G$ are
\begin{equation*}
G_p\left(p,l\right)=\pi_p\left(p,l\right)+\frac{l}{1-\alpha},
\end{equation*}
and
\begin{equation*}
G_l\left(p,l\right)=\pi_l\left(p,l\right)+\frac{p}{1-\alpha}.
\end{equation*}
The corresponding first-order isoquants are
\begin{equation*}
\mathcal{I}_{\text{F}}=\left\{\left(p,l\right)\in\Theta^{\left(\text{L}\right)}:G_p\left(p,l\right)=\pi_p\left(p,l\right)+\frac{l}{1-\alpha}=0\right\},
\end{equation*}
and
\begin{equation*}
\mathcal{I}_{\text{S}}=\left\{\left(p,l\right)\in\Theta^{\left(\text{L}\right)}:G_l\left(p,l\right)=\pi_l\left(p,l\right)+\frac{p}{1-\alpha}=0\right\}.
\end{equation*}
The isoquant $\mathcal{I}_{\text{F}}$ consists of points at which the marginal effect of changing the residual loss probability is zero, while $\mathcal{I}_{\text{S}}$ consists of points at which the marginal effect of changing the residual loss severity is zero. In other words, $\mathcal{I}_{\text{F}}$ is the marginal-balance curve for self-protection, and $\mathcal{I}_{\text{S}}$ is the marginal-balance curve for self-insurance.

The geometry of these two isoquants is governed by the cross-partial derivative of the left-region TVaR objective:
\begin{equation}
G_{pl}\left(p,l\right)=\pi_{pl}\left(p,l\right)+\frac{1}{1-\alpha}.
\label{eq:G_pl}
\end{equation}
Since $\pi_{pl}=c_{xy}$, the cross-partial derivative can be interpreted as the net interaction between residual frequency and residual severity. The component $c_{xy}$ captures the technological interaction in the joint risk-reduction cost, whereas $1/\left(1-\alpha\right)$ comes from the TVaR tail-risk term $pl/\left(1-\alpha\right)$. Therefore, $G_{pl}$ is the net interaction term, combining the cost-side technological interaction and the risk-measure-side interaction.

The second-order derivatives of $G$ are
\begin{equation*}
G_{pp}\left(p,l\right)=\pi_{pp}\left(p,l\right)>0,
\end{equation*}
and
\begin{equation*}
G_{ll}\left(p,l\right)=\pi_{ll}\left(p,l\right)>0,
\end{equation*}
where the positivity follows from the maintained own-curvature assumptions. Together with the sign of $G_{pl}$, these curvatures determine the qualitative shapes of the isoquants.

Suppose first that $G_{pl}\left(p,l\right)>0$ on $\Theta^{\left(\text{L}\right)}$. Where the isoquants can be represented as graphs, their slopes are
\begin{equation}
\frac{\text{d}l}{\text{d}p}\left(p,l\right)\bigg\vert_{\left(p,l\right)\in\mathcal{I}_{\text{F}}}=-\frac{\pi_{pp}\left(p,l\right)}{\pi_{pl}\left(p,l\right)+\frac{1}{1-\alpha}}=-\frac{G_{pp}\left(p,l\right)}{G_{pl}\left(p,l\right)}<0,
\label{eq:mathcal_I_F_negative}
\end{equation}
and
\begin{equation}
\frac{\text{d}l}{\text{d}p}\left(p,l\right)\bigg\vert_{\left(p,l\right)\in\mathcal{I}_{\text{S}}}=-\frac{\pi_{pl}\left(p,l\right)+\frac{1}{1-\alpha}}{\pi_{ll}\left(p,l\right)}=-\frac{G_{pl}\left(p,l\right)}{G_{ll}\left(p,l\right)}<0.
\label{eq:mathcal_I_S_negative}
\end{equation}
Therefore, under positive net interaction, both first-order isoquants are downward-sloping. Moreover, since $G_{pp}\left(p,l\right)>0$, $G_{pl}\left(p,l\right)>0$, and $G_{ll}\left(p,l\right)>0$, both first-order derivatives $G_p$ and $G_l$ increase when moving northeast and decrease when moving southwest in the $\left(p,l\right)$-plane. In this case, the downward-sloping isoquants generate a lower-left/upper-right separation geometry.

Suppose next that $G_{pl}\left(p,l\right)<0$ on $\Theta^{\left(\text{L}\right)}$. Where the isoquants can be represented as graphs, their slopes are
\begin{equation*}
\frac{\text{d}l}{\text{d}p}\left(p,l\right)\bigg\vert_{\left(p,l\right)\in\mathcal{I}_{\text{F}}}=-\frac{G_{pp}\left(p,l\right)}{G_{pl}\left(p,l\right)}>0,
\end{equation*}
and
\begin{equation*}
\frac{\text{d}l}{\text{d}p}\left(p,l\right)\bigg\vert_{\left(p,l\right)\in\mathcal{I}_{\text{S}}}=-\frac{G_{pl}\left(p,l\right)}{G_{ll}\left(p,l\right)}>0.
\end{equation*}
Therefore, under negative net interaction, both first-order isoquants are upward-sloping. The directional monotonicities of the two marginal-balance curves are now opposite. Since $G_{pp}\left(p,l\right)>0$ and $G_{lp}\left(p,l\right)<0$, the first-order derivative $G_p$ increases when moving southeast and decreases when moving northwest in the $\left(p,l\right)$-plane. Conversely, since $G_{pl}\left(p,l\right)<0$ and $G_{ll}\left(p,l\right)>0$, the first-order derivative $G_l$ increases when moving northwest and decreases when moving southeast in the plane. In this case, the separation geometry is reversed; the upward-sloping isoquants generate an upper-left/lower-right separation structure.

Finally, suppose that $G_{pl}\left(p,l\right)=0$ on $\Theta^{\left(\text{L}\right)}$. This is a degenerate case in which the cost-side interaction exactly offsets the TVaR tail-risk interaction, so that $c_{xy}=-1/\left(1-\alpha\right)$. In this case, the left-region objective has zero cross-partial derivative. The isoquant $\mathcal{I}_{\text{F}}$ is vertical in the $\left(p,l\right)$-plane, while the isoquant $\mathcal{I}_{\text{S}}$ is horizontal in the plane. Thus the marginal-balance geometry becomes coordinatewise rather than upward- or downward-sloping.

The following analysis focuses on cases in which the sign of $G_{pl}$ is constant on $\Theta^{\left(\text{L}\right)}$. If $G_{pl}$ changes sign within $\Theta^{\left(\text{L}\right)}$, then the isoquants may change their monotonicity inside the feasible rectangle, leading to a substantially more complicated constrained geometry. Such mixed-sign cases are mathematically interesting, but they are left for future work.

\subsection{Isoquant Analysis}

This subsection studies the left-region TVaR problem \eqref{eq:TVaR_problem<=} when $\alpha\in\left[0,1\right)$ and $\Theta^{\left(\text{L}\right)}\neq\emptyset$. Define $\overline{p}_{\alpha}=\min\left\{\overline{p},1-\alpha\right\}$. Then
\begin{equation*}
\Theta^{\left(\text{L}\right)}=\left[\underline{p},\overline{p}_{\alpha}\right]\times\left[\underline{l},\overline{l}\right].
\end{equation*}

Since all sign regions in this subsection are subsets of the left-region feasible set $\Theta^{\left(\text{L}\right)}$, we suppress the superscript L for notational simplicity. Define
\begin{equation*}
\Theta_{\text{F}}^{\left(<\right)}=\left\{\left(p,l\right)\in\Theta^{\left(\text{L}\right)}:G_p\left(p,l\right)<0\right\},
\end{equation*}
\begin{equation*}
\Theta_{\text{F}}^{\left(=\right)}=\left\{\left(p,l\right)\in\Theta^{\left(\text{L}\right)}:G_p\left(p,l\right)=0\right\},
\end{equation*}
and
\begin{equation*}
\Theta_{\text{F}}^{\left(>\right)}=\left\{\left(p,l\right)\in\Theta^{\left(\text{L}\right)}:G_p\left(p,l\right)>0\right\}.
\end{equation*}
Similarly, define
\begin{equation*}
\Theta_{\text{S}}^{\left(<\right)}=\left\{\left(p,l\right)\in\Theta^{\left(\text{L}\right)}:G_l\left(p,l\right)<0\right\},
\end{equation*}
\begin{equation*}
\Theta_{\text{S}}^{\left(=\right)}=\left\{\left(p,l\right)\in\Theta^{\left(\text{L}\right)}:G_l\left(p,l\right)=0\right\},
\end{equation*}
and
\begin{equation*}
\Theta_{\text{S}}^{\left(>\right)}=\left\{\left(p,l\right)\in\Theta^{\left(\text{L}\right)}:G_l\left(p,l\right)>0\right\}.
\end{equation*}
The corresponding weak sign regions are defined in the obvious way; for example, $\Theta_{\text{F}}^{\left(\leq\right)}=\Theta_{\text{F}}^{\left(<\right)}\cup\Theta_{\text{F}}^{\left(=\right)}$ and $\Theta_{\text{F}}^{\left(\geq\right)}=\Theta_{\text{F}}^{\left(=\right)}\cup\Theta_{\text{F}}^{\left(>\right)}$. The sets $\Theta_{\text{S}}^{\left(\leq\right)}$ and $\Theta_{\text{S}}^{\left(\geq\right)}$ are defined similarly.

The sign of $G_p$ indicates the local direction in which changing the residual loss probability reduces the left-region TVaR objective; if $G_p<0$, increasing $p$ locally reduces $G$, whereas if $G_p>0$, decreasing $p$ locally reduces $G$. Similarly, the sign of $G_l$ indicates the local direction in which changing the residual loss severity reduces the objective. Hence, the global solution of the left-region problem can be studied by examining how the two first-order isoquants $\mathcal{I}_{\text{F}}=\Theta_{\text{F}}^{\left(=\right)}$ and $\mathcal{I}_{\text{S}}=\Theta_{\text{S}}^{\left(=\right)}$ divide the rectangular feasible set $\Theta^{\left(\text{L}\right)}$.

\subsubsection{Positive Net Interaction}\label{sec:positive_net}

Throughout this subsection, assume that $G_{pl}>0$ on $\Theta^{\left(\text{L}\right)}$. As discussed in Section \ref{sec:iso_geo}, under this positive net interaction condition, the first-order isoquants $\mathcal{I}_{\text{F}}=\Theta_{\text{F}}^{\left(=\right)}$ and $\mathcal{I}_{\text{S}}=\Theta_{\text{S}}^{\left(=\right)}$ are downward-sloping whenever they are non-empty and can be represented as graphs.  Moreover, the first-order derivatives $G_p$ and $G_l$ are increasing in both residual-risk variables. This gives rise to a lower-left/upper-right separation geometry. The following lemma formalizes this property.

\begin{lemma}\label{eq:gradient_property}
Assume that $\Theta^{\left(\text{L}\right)}\neq\emptyset$ and $G_{pl}>0$ on $\Theta^{\left(\text{L}\right)}$. Then the following statements hold.
\begin{enumerate}
\item[(i)] For any $\left(p,l\right)\in\Theta_{\text{F}}^{\left(\leq\right)}$, and for any $\left(\tilde{p},\tilde{l}\right)\in\Theta^{\left(\text{L}\right)}\setminus\left\{\left(p,l\right)\right\}$ with $\tilde{p}\leq p$ and $\tilde{l}\leq l$, we have $\left(\tilde{p},\tilde{l}\right)\in\Theta_{\text{F}}^{\left(<\right)}$.
\item[(ii)] For any $\left(p,l\right)\in\Theta_{\text{F}}^{\left(\geq\right)}$, and for any $\left(\tilde{p},\tilde{l}\right)\in\Theta^{\left(\text{L}\right)}\setminus\left\{\left(p,l\right)\right\}$ with $\tilde{p}\geq p$ and $\tilde{l}\geq l$, we have $\left(\tilde{p},\tilde{l}\right)\in\Theta_{\text{F}}^{\left(>\right)}$.
\item[(iii)] For any $\left(p,l\right)\in\Theta_{\text{S}}^{\left(\leq\right)}$, and for any $\left(\tilde{p},\tilde{l}\right)\in\Theta^{\left(\text{L}\right)}\setminus\left\{\left(p,l\right)\right\}$ with $\tilde{p}\leq p$ and $\tilde{l}\leq l$, we have $\left(\tilde{p},\tilde{l}\right)\in\Theta_{\text{S}}^{\left(<\right)}$.
\item[(iv)] For any $\left(p,l\right)\in\Theta_{\text{S}}^{\left(\geq\right)}$, and for any $\left(\tilde{p},\tilde{l}\right)\in\Theta^{\left(\text{L}\right)}\setminus\left\{\left(p,l\right)\right\}$ with $\tilde{p}\geq p$ and $\tilde{l}\geq l$, we have $\left(\tilde{p},\tilde{l}\right)\in\Theta_{\text{S}}^{\left(>\right)}$.
\end{enumerate}
Consequently, if $\mathcal{I}_{\text{F}}=\Theta_{\text{F}}^{\left(=\right)}$ or $\mathcal{I}_{\text{S}}=\Theta_{\text{S}}^{\left(=\right)}$ is non-empty and can be represented as a graph, then it is downward-sloping. When such an isoquant cuts through the rectangle $\Theta^{\left(\text{L}\right)}$, it separates the rectangle into lower-left and upper-right sign regions.
\end{lemma}
\begin{proof}
We first prove the statements involving  $\Theta_{\text{F}}^{\left(\leq\right)}$ and $\Theta_{\text{F}}^{\left(\geq\right)}$.

Since $G_{pp}\left(p,l\right)=\pi_{pp}\left(p,l\right)>0$ and $G_{pl}\left(p,l\right)>0$, the first-order derivative $G_p$ is strictly increasing in both $p$ and $l$. Therefore, if $G_{p}\left(p,l\right)\leq 0$ and $\left(\tilde{p},\tilde{l}\right)\neq\left(p,l\right)$ satisfies $\tilde{p}\leq p$ and $\tilde{l}\leq l$, with at least one inequality strict, then $G_{p}\left(\tilde{p},\tilde{l}\right)<G_{p}\left(p,l\right)\leq 0$. Hence, $\left(\tilde{p},\tilde{l}\right)\in\Theta_{\text{F}}^{\left(<\right)}$. This proves (i).

Similarly, if $G_{p}\left(p,l\right)\geq 0$ and $\left(\tilde{p},\tilde{l}\right)\neq\left(p,l\right)$ satisfies $\tilde{p}\geq p$ and $\tilde{l}\geq l$, with at least one inequality strict, then $G_{p}\left(\tilde{p},\tilde{l}\right)>G_{p}\left(p,l\right)\geq 0$. Hence, $\left(\tilde{p},\tilde{l}\right)\in\Theta_{\text{F}}^{\left(>\right)}$. This proves (ii).

For the severity derivative, note that $G_{lp}\left(p,l\right)=G_{pl}\left(p,l\right)>0$ and $G_{ll}\left(p,l\right)=\pi_{ll}\left(p,l\right)>0$. Thus $G_l$ is also strictly increasing in both $p$ and $l$. The same argument proves (iii) and (iv). The final assertion follows from the slope formulas \eqref{eq:mathcal_I_F_negative} and \eqref{eq:mathcal_I_S_negative}.
\end{proof}

The first case arises when at least one of the two first-order isoquants does not cut through the rectangle $\Theta^{\left(\text{L}\right)}$. Then at least one of the first-order derivatives has a constant sign on $\Theta^{\left(\text{L}\right)}$, and the two-dimensional problem reduces to a one-dimensional boundary problem.

For $p_0\in\left[\underline{p},\overline{p}_{\alpha}\right]$, define
\begin{equation}
l^*\left(p_0\right)\in\argmin_{l\in\left[\underline{l},\overline{l}\right]}G\left(p_0,l\right),
\label{eq:boundary_l^*}
\end{equation}
and, for $l_0\in\left[\underline{l},\overline{l}\right]$, define
\begin{equation}
p^*\left(l_0\right)\in\argmin_{p\in\left[\underline{p},\overline{p}_{\alpha}\right]}G\left(p,l_0\right).
\label{eq:boundary_p^*}
\end{equation}
Since $G_{ll}>0$ and $G_{pp}>0$ on $\Theta^{\left(\text{L}\right)}$, these one-dimensional minimizers are unique.

\begin{proposition}\label{prop:tvar_first_result}
Assume that $\Theta^{\left(\text{L}\right)}\neq\emptyset$ and $G_{pl}>0$ on $\Theta^{\left(\text{L}\right)}$. Then the following statements hold.

\begin{enumerate}
\item[(i)] If $G_p\left(\overline{p}_{\alpha},\overline{l}\right)\leq 0$, then the optimal solution of the left-region problem \eqref{eq:TVaR_problem<=} is ${\bm\theta}^{*\left(\text{L}\right)}=\left(p^{*\left(\text{L}\right)},l^{*\left(\text{L}\right)}\right)=\left(\overline{p}_{\alpha},l^*\left(\overline{p}_{\alpha}\right)\right)$.

\item[(ii)] If $G_p\left(\underline{p},\underline{l}\right)\geq 0$, then the optimal solution of the left-region problem \eqref{eq:TVaR_problem<=} is ${\bm\theta}^{*\left(\text{L}\right)}=\left(\underline{p},l^*\left(\underline{p}\right)\right)$.

\item[(iii)] If $G_l\left(\overline{p}_{\alpha},\overline{l}\right)\leq 0$, then the optimal solution of the left-region problem \eqref{eq:TVaR_problem<=} is ${\bm\theta}^{*\left(\text{L}\right)}=\left(p^*\left(\overline{l}\right),\overline{l}\right)$.

\item[(iv)] If $G_l\left(\underline{p},\underline{l}\right)\geq 0$, then the optimal solution of the left-region problem \eqref{eq:TVaR_problem<=} is ${\bm\theta}^{*\left(\text{L}\right)}=\left(p^*\left(\underline{l}\right),\underline{l}\right)$.
\end{enumerate}
\end{proposition}
\begin{proof}
We prove (i). Suppose that $G_p\left(\overline{p}_{\alpha},\overline{l}\right)\leq 0$. Then $\left(\overline{p}_{\alpha},\overline{l}\right)\in\Theta_{\text{F}}^{\left(\leq\right)}$. By Lemma \ref{eq:gradient_property}, for any $\left(p,l\right)\in\Theta^{\left(\text{L}\right)}\setminus\left\{\left(\overline{p}_{\alpha},\overline{l}\right)\right\}$, we have $\left(p,l\right)\in\Theta_{\text{F}}^{\left(<\right)}$. Therefore, $G_p\left(p,l\right)<0$ for all $\left(p,l\right)\in\Theta^{\left(\text{L}\right)}\setminus\left\{\left(\overline{p}_{\alpha},\overline{l}\right)\right\}$. In particular, for each fixed $l\in\left[\underline{l},\overline{l}\right]$, $G_p\left(p,l\right)<0$ for all $p<\overline{p}_{\alpha}$. Hence $p\mapsto G\left(p,l\right)$ is strictly decreasing on $\left[\underline{p},\overline{p}_{\alpha}\right]$. Thus any optimal solution must have $p^{*\left(\text{L}\right)}=\overline{p}_{\alpha}$. The left-region problem therefore reduces to the one-dimensional problem
\begin{equation*}
\inf_{l\in\left[\underline{l},\overline{l}\right]}G\left(\overline{p}_{\alpha},l\right),
\end{equation*}
whose unique minimizer is $l^*\left(\overline{p}_{\alpha}\right)$. Hence, ${\bm\theta}^{*\left(\text{L}\right)}=\left(\overline{p}_{\alpha},l^*\left(\overline{p}_{\alpha}\right)\right)$. This proves (i).

The proofs of (ii)--(iv) are analogous.
\end{proof}

Proposition \ref{prop:tvar_first_result} corresponds to the case in which at least one marginal-balance isoquant does not cut through the feasible rectangle. Economically, one of the two risk-reduction margins has a uniform local direction of improvement over the entire left-region problem. Hence the RH is pushed to a boundary in one residual-risk dimension. The remaining choice is then a one-dimensional tradeoff along that boundary. In this case, the optimal strategy is governed by a single constrained marginal balance. Either the RH fixes residual frequency at its left or right boundary and chooses residual severity optimally, or fixes residual severity at its lower or upper boundary and chooses residual frequency optimally.

We next consider the case in which both first-order isoquants cut through the rectangle $\Theta^{\left(\text{L}\right)}$, but they do not intersect. By ``cutting through'', we mean that each isoquant is non-empty and separates $\Theta^{\left(\text{L}\right)}$ into two non-empty sign regions. By ``not intersecting'', we mean that the isoquants neither touch nor cross each other. Equivalently, $\Theta_{\text{F}}^{\left(<\right)}$, $\Theta_{\text{F}}^{\left(=\right)}=\mathcal{I}_{\text{F}}$, $\Theta_{\text{F}}^{\left(>\right)}$, $\Theta_{\text{S}}^{\left(<\right)}$, $\Theta_{\text{S}}^{\left(=\right)}=\mathcal{I}_{\text{S}}$, and $\Theta_{\text{S}}^{\left(>\right)}$ are all non-empty, but $\Theta_{\text{F}}^{\left(=\right)}\cap\Theta_{\text{S}}^{\left(=\right)}=\mathcal{I}_{\text{F}}\cap\mathcal{I}_{\text{S}}=\emptyset$. Since the two isoquants do not intersect, the frequency marginal-balance isoquant $\Theta_{\text{F}}^{\left(=\right)}=\mathcal{I}_{\text{F}}$ lies entirely in either $\Theta_{\text{S}}^{\left(<\right)}$ or $\Theta_{\text{S}}^{\left(>\right)}$. Equivalently, the severity marginal-balance isoquant $\Theta_{\text{S}}^{\left(=\right)}=\mathcal{I}_{\text{S}}$ lies entirely in either $\Theta_{\text{F}}^{\left(>\right)}$ or $\Theta_{\text{F}}^{\left(<\right)}$, respectively. These two configurations lead to different extreme candidates.

To describe the possible configurations, define the joint sign regions
\begin{equation*}
\Theta_{\text{F,S}}^{\left(<,<\right)}=\Theta_{\text{F}}^{\left(<\right)}\cap\Theta_{\text{S}}^{\left(<\right)}=\left\{\left(p,l\right)\in\Theta^{\left(\text{L}\right)}:G_p\left(p,l\right)<0,\;G_l\left(p,l\right)<0\right\},
\end{equation*}
\begin{equation*}
\Theta_{\text{F,S}}^{\left(<,=\right)}=\Theta_{\text{F}}^{\left(<\right)}\cap\Theta_{\text{S}}^{\left(=\right)}=\left\{\left(p,l\right)\in\Theta^{\left(\text{L}\right)}:G_p\left(p,l\right)<0,\;G_l\left(p,l\right)=0\right\},
\end{equation*}
\begin{equation*}
\Theta_{\text{F,S}}^{\left(<,>\right)}=\Theta_{\text{F}}^{\left(<\right)}\cap\Theta_{\text{S}}^{\left(>\right)}=\left\{\left(p,l\right)\in\Theta^{\left(\text{L}\right)}:G_p\left(p,l\right)<0,\;G_l\left(p,l\right)>0\right\},
\end{equation*}
and define $\Theta_{\text{F,S}}^{\left(=,<\right)}$, $\Theta_{\text{F,S}}^{\left(=,=\right)}$, $\Theta_{\text{F,S}}^{\left(=,>\right)}$, $\Theta_{\text{F,S}}^{\left(>,<\right)}$, $\Theta_{\text{F,S}}^{\left(>,=\right)}$, $\Theta_{\text{F,S}}^{\left(>,>\right)}$ similarly. In particular,
\begin{equation*}
\Theta_{\text{F,S}}^{\left(=,=\right)}=\Theta_{\text{F}}^{\left(=\right)}\cap\Theta_{\text{S}}^{\left(=\right)}=\mathcal{I}_{\text{F}}\cap\mathcal{I}_{\text{S}}=\left\{\left(p,l\right)\in\Theta^{\left(\text{L}\right)}:G_p\left(p,l\right)=0,\;G_l\left(p,l\right)=0\right\}.
\end{equation*}

Using the one-dimensional minimizers defined in \eqref{eq:boundary_l^*} and \eqref{eq:boundary_p^*}, define the upper-left and lower-right candidates by
\begin{equation}
{\bm\theta}^{\text{UL}}=\left(p^{\text{UL}},l^{\text{UL}}\right)=\left(p^*\left(\overline{l}\right),l^*\left(\underline{p}\right)\right),
\label{eq:theta_UL}
\end{equation}
and
\begin{equation}
{\bm\theta}^{\text{LR}}=\left(p^{\text{LR}},l^{\text{LR}}\right)=\left(p^*\left(\underline{l}\right),l^*\left(\overline{p}_{\alpha}\right)\right).
\label{eq:theta_LR}
\end{equation}
The candidate ${\bm\theta}^{\text{UL}}$ combines the upper-boundary constrained marginal balance for $p$ with the left-boundary constrained marginal balance for $l$, whereas ${\bm\theta}^{\text{LR}}$ combines the lower-boundary constrained marginal balance for $p$ with the right-boundary constrained marginal balance for $l$.

\begin{proposition}\label{prop:tvar_second_result}
Assume that $\Theta^{\left(\text{L}\right)}\neq\emptyset$ and $G_{pl}>0$ on $\Theta^{\left(\text{L}\right)}$. Suppose that $G_p\left(\underline{p},\underline{l}\right)<0<G_p\left(\overline{p}_{\alpha},\overline{l}\right)$ and $G_l\left(\underline{p},\underline{l}\right)<0<G_l\left(\overline{p}_{\alpha},\overline{l}\right)$, but $\Theta_{\text{F,S}}^{\left(=,=\right)}=\emptyset$. Then the following statements hold.

\begin{enumerate}
\item[(i)] If $\Theta_{\text{F}}^{\left(=\right)}\subseteq\Theta_{\text{S}}^{\left(<\right)}$ or, equivalently, $\Theta_{\text{S}}^{\left(=\right)}\subseteq\Theta_{\text{F}}^{\left(>\right)}$, then the optimal solution of the left-region problem \eqref{eq:TVaR_problem<=} is ${\bm\theta}^{*\left(\text{L}\right)}={\bm\theta}^{\text{UL}}=\left(p^*\left(\overline{l}\right),l^*\left(\underline{p}\right)\right)$.
\item[(ii)] If $\Theta_{\text{F}}^{\left(=\right)}\subseteq\Theta_{\text{S}}^{\left(>\right)}$ or, equivalently, $\Theta_{\text{S}}^{\left(=\right)}\subseteq\Theta_{\text{F}}^{\left(<\right)}$, then the optimal solution of the left-region problem \eqref{eq:TVaR_problem<=} is ${\bm\theta}^{*\left(\text{L}\right)}={\bm\theta}^{\text{LR}}=\left(p^*\left(\underline{l}\right),l^*\left(\overline{p}_{\alpha}\right)\right)$.
\end{enumerate}
\end{proposition}

\begin{proof}
We prove (i). Suppose that $\Theta_{\text{F}}^{\left(=\right)}\subseteq\Theta_{\text{S}}^{\left(<\right)}$. Equivalently, $\Theta_{\text{S}}^{\left(=\right)}\subseteq\Theta_{\text{F}}^{\left(>\right)}$. Under this configuration, the frequency marginal-balance isoquant lies entirely in the region where the severity derivative is negative, while the severity marginal-balance isoquant lies entirely in the region where the frequency derivative is positive. Therefore, the feasible rectangle is decomposed as
\begin{equation*}
\Theta^{\left(\text{L}\right)}=\Theta_{\text{F,S}}^{\left(<,<\right)}\cup\Theta_{\text{F,S}}^{\left(=,<\right)}\cup\Theta_{\text{F,S}}^{\left(>,<\right)}\cup\Theta_{\text{F,S}}^{\left(>,=\right)}\cup\Theta_{\text{F,S}}^{\left(>,>\right)}.
\end{equation*}
Indeed, the joint sign regions $\Theta_{\text{F,S}}^{\left(<,=\right)}$, $\Theta_{\text{F,S}}^{\left(<,>\right)}$, $\Theta_{\text{F,S}}^{\left(=,=\right)}$, $\Theta_{\text{F,S}}^{\left(=,>\right)}$ are empty under the assumed configuration.

We first show that no optimal solution can lie in the two regions $\Theta_{\text{F,S}}^{\left(<,<\right)}$ and $\Theta_{\text{F,S}}^{\left(>,>\right)}$. Let $\left(p,l\right)\in\Theta_{\text{F,S}}^{\left(<,<\right)}$. Then $G_p\left(p,l\right)<0$ and $G_l\left(p,l\right)<0$. By Lemma \ref{eq:gradient_property}, there exists a feasible northeast movement from $\left(p,l\right)$ along which the objective strictly decreases until the path reaches the boundary of $\Theta_{\text{F,S}}^{\left(=,<\right)}\cup\Theta_{\text{F,S}}^{\left(>,<\right)}\cup\Theta_{\text{F,S}}^{\left(>,=\right)}$. Along such a movement, the directional derivative of $G$ is negative. Hence every strategy in $\Theta_{\text{F,S}}^{\left(<,<\right)}$ is strictly dominated by some strategy in $\Theta_{\text{F,S}}^{\left(=,<\right)}\cup\Theta_{\text{F,S}}^{\left(>,<\right)}\cup\Theta_{\text{F,S}}^{\left(>,=\right)}$. Thus no strategy in $\Theta_{\text{F,S}}^{\left(<,<\right)}$ can be optimal. Similar arguments show that every strategy in $\Theta_{\text{F,S}}^{\left(>,>\right)}$ is strictly dominated by some strategy in $\Theta_{\text{F,S}}^{\left(=,<\right)}\cup\Theta_{\text{F,S}}^{\left(>,<\right)}\cup\Theta_{\text{F,S}}^{\left(>,=\right)}$, and thus no strategy in $\Theta_{\text{F,S}}^{\left(>,>\right)}$ can be optimal.

Next, we show that ${\bm\theta}^{\text{UL}}=\left(p^{\text{UL}},l^{\text{UL}}\right)\in\Theta_{\text{F,S}}^{\left(=,<\right)}\cup\Theta_{\text{F,S}}^{\left(>,<\right)}\cup\Theta_{\text{F,S}}^{\left(>,=\right)}$, and that every strategy in this set lies southeast of ${\bm\theta}^{\text{UL}}$. Under the assumed configuration, there are the following three cases.
\begin{enumerate}
\item[(I)] Suppose that $\left(\underline{p},\overline{l}\right)\in\Theta_{\text{F}}^{\left(\leq\right)}$. Then, under the present ordering, $\left(\underline{p},\overline{l}\right)\in\Theta_{\text{S}}^{\left(<\right)}$. In this case, ${\bm\theta}^{\text{UL}}=\left(p^{\text{UL}},l^{\text{UL}}\right)=\left(p^*\left(\overline{l}\right),\overline{l}\right)$. Moreover, $p^*\left(\overline{l}\right)\in\left[\underline{p},\overline{p}_{\alpha}\right)$ is the constrained marginal-balance strategy on the upper boundary; if it is interior, then it satisfies $G_p\left(p^*\left(\overline{l}\right),\overline{l}\right)=0$. Thus ${\bm\theta}^{\text{UL}}\in\Theta_{\text{F,S}}^{\left(=,<\right)}\subseteq\Theta_{\text{F,S}}^{\left(=,<\right)}\cup\Theta_{\text{F,S}}^{\left(>,<\right)}\cup\Theta_{\text{F,S}}^{\left(>,=\right)}$. Let $\left(p,l\right)\in\Theta_{\text{F,S}}^{\left(=,<\right)}\cup\Theta_{\text{F,S}}^{\left(>,<\right)}\cup\Theta_{\text{F,S}}^{\left(>,=\right)}$. Clearly, $l\leq\overline{l}=l^{\text{UL}}$. If $p^{\text{UL}}=\underline{p}$, then $p^{\text{UL}}\leq p$. If $p^{\text{UL}}>\underline{p}$, suppose, to the contrary, that $p<p^{\text{UL}}$. Since ${\bm\theta}^{\text{UL}}\in\Theta_{\text{F,S}}^{\left(=,<\right)}\subseteq\Theta_{\text{F}}^{\left(=\right)}$, as well as $p<p^{\text{UL}}$ and $l\leq l^{\text{UL}}$, Lemma \ref{eq:gradient_property} implies that $\left(p,l\right)\in\Theta_{\text{F}}^{\left(<\right)}$, which contradicts $\left(p,l\right)\in\Theta_{\text{F,S}}^{\left(=,<\right)}\cup\Theta_{\text{F,S}}^{\left(>,<\right)}\cup\Theta_{\text{F,S}}^{\left(>,=\right)}$. Hence $p^{\text{UL}}\leq p$.
\item[(II)] When $\left(\underline{p},\overline{l}\right)\in\Theta_{\text{F,S}}^{\left(>,<\right)}$, ${\bm\theta}^{\text{UL}}=\left(p^{\text{UL}},l^{\text{UL}}\right)=\left(\underline{p},\overline{l}\right)$, and hence ${\bm\theta}^{\text{UL}}\in\Theta_{\text{F,S}}^{\left(>,<\right)}$. For any $\left(p,l\right)\in\Theta_{\text{F,S}}^{\left(=,<\right)}\cup\Theta_{\text{F,S}}^{\left(>,<\right)}\cup\Theta_{\text{F,S}}^{\left(>,=\right)}$, we have immediately $p^{\text{UL}}=\underline{p}\leq p$ and $l\leq\overline{l}=l^{\text{UL}}$.
\item[(III)] Suppose that $\left(\underline{p},\overline{l}\right)\in\Theta_{\text{S}}^{\left(\geq\right)}$. Then, under the present ordering, $\left(\underline{p},\overline{l}\right)\in\Theta_{\text{F}}^{\left(>\right)}$. In this case, ${\bm\theta}^{\text{UL}}=\left(p^{\text{UL}},l^{\text{UL}}\right)=\left(\underline{p},l^*\left(\underline{p}\right)\right)$. Moreover, $l^*\left(\underline{p}\right)\in\left(\underline{l},\overline{l}\right]$ is the constrained marginal-balance strategy on the left boundary; if it is interior, then it satisfies $G_l\left(\underline{p},l^*\left(\underline{p}\right)\right)=0$. Thus ${\bm\theta}^{\text{UL}}\in\Theta_{\text{F,S}}^{\left(>,=\right)}\subseteq\Theta_{\text{F,S}}^{\left(=,<\right)}\cup\Theta_{\text{F,S}}^{\left(>,<\right)}\cup\Theta_{\text{F,S}}^{\left(>,=\right)}$. Let $\left(p,l\right)\in\Theta_{\text{F,S}}^{\left(=,<\right)}\cup\Theta_{\text{F,S}}^{\left(>,<\right)}\cup\Theta_{\text{F,S}}^{\left(>,=\right)}$. Clearly, $p^{\text{UL}}=\underline{p}\leq p$. If $l^{\text{UL}}=\overline{l}$, then $l\leq l^{\text{UL}}$. If $l^{\text{UL}}<\overline{l}$, suppose, to the contrary, that $l>l^{\text{UL}}$. Since ${\bm\theta}^{\text{UL}}\in\Theta_{\text{F,S}}^{\left(>,=\right)}\subseteq\Theta_{\text{S}}^{\left(=\right)}$, as well as $p^{\text{UL}}\leq p$ and $l^{\text{UL}}<l$, Lemma \ref{eq:gradient_property} implies that $\left(p,l\right)\in\Theta_{\text{S}}^{\left(>\right)}$, which contradicts $\left(p,l\right)\in\Theta_{\text{F,S}}^{\left(=,<\right)}\cup\Theta_{\text{F,S}}^{\left(>,<\right)}\cup\Theta_{\text{F,S}}^{\left(>,=\right)}$. Hence $l\leq l^{\text{UL}}$.
\end{enumerate}

Combining the three cases, for every $\left(\tilde{p},\tilde{l}\right)\in\left(\Theta_{\text{F,S}}^{\left(=,<\right)}\cup\Theta_{\text{F,S}}^{\left(>,<\right)}\cup\Theta_{\text{F,S}}^{\left(>,=\right)}\right)\setminus\left\{{\bm\theta}^{\text{UL}}\right\}$, we have $p^{\text{UL}}\leq\tilde{p}$ and $\tilde{l}\leq l^{\text{UL}}$, with at least one inequality strict. Moreover, throughout $\Theta_{\text{F,S}}^{\left(=,<\right)}\cup\Theta_{\text{F,S}}^{\left(>,<\right)}\cup\Theta_{\text{F,S}}^{\left(>,=\right)}$, we have $G_p\geq 0$ and $G_l\leq 0$. By the separation geometry in Lemma \ref{eq:gradient_property}, there exists a monotone path from $\left(\tilde{p},\tilde{l}\right)$ to ${\bm\theta}^{\text{UL}}$, contained in $\Theta_{\text{F,S}}^{\left(=,<\right)}\cup\Theta_{\text{F,S}}^{\left(>,<\right)}\cup\Theta_{\text{F,S}}^{\left(>,=\right)}$, along which $p$ decreases and $l$ increases, with at least one coordinate changing strictly. Along this path, $\text{d}G=G_p\text{d}p+G_l\text{d}l\leq 0$, with strict inequality on a set of positive path length. Hence, by the gradient theorem, $G\left({\bm\theta}^{\text{UL}}\right)<G\left(\tilde{p},\tilde{l}\right)$. Consequently, ${\bm\theta}^{*\left(\text{L}\right)}={\bm\theta}^{\text{UL}}=\left(p^*\left(\overline{l}\right),l^*\left(\underline{p}\right)\right)$. This proves (i).

The proof of (ii) is analogous, with the roles of the upper-left and lower-right candidates reversed. Under the ordering $\Theta_{\text{F}}^{\left(=\right)}\subseteq\Theta_{\text{S}}^{\left(>\right)}$ or, equivalently, $\Theta_{\text{S}}^{\left(=\right)}\subseteq\Theta_{\text{F}}^{\left(<\right)}$, the joint sign regions that an optimal solution can lie in, after eliminating $\Theta_{\text{F,S}}^{\left(<,<\right)}$ and $\Theta_{\text{F,S}}^{\left(>,>\right)}$, are $\Theta_{\text{F,S}}^{\left(<,=\right)}\cup\Theta_{\text{F,S}}^{\left(<,>\right)}\cup\Theta_{\text{F,S}}^{\left(=,>\right)}$. The same argument shows that every strategy in this set lies northwest of ${\bm\theta}^{\text{LR}}$, and that moving southeast toward ${\bm\theta}^{\text{LR}}$ decreases $G$, with strict decrease unless the starting point is ${\bm\theta}^{\text{LR}}$ itself in the set. Hence ${\bm\theta}^{*\left(\text{L}\right)}={\bm\theta}^{\text{LR}}=\left(p^*\left(\underline{l}\right),l^*\left(\overline{p}_{\alpha}\right)\right)$. This proves (ii).
\end{proof}

Under the assumptions of Proposition \ref{prop:tvar_second_result}, the inequalities $G_p\left(\underline{p},\underline{l}\right)<0<G_p\left(\overline{p}_{\alpha},\overline{l}\right)$ and $G_l\left(\underline{p},\underline{l}\right)<0<G_l\left(\overline{p}_{\alpha},\overline{l}\right)$ ensure that both first-order isoquants cut through the rectangle $\Theta^{\left(\text{L}\right)}$. Equivalently, the sign regions $\Theta_{\text{F}}^{\left(<\right)}$, $\Theta_{\text{F}}^{\left(=\right)}$, $\Theta_{\text{F}}^{\left(>\right)}$, $\Theta_{\text{S}}^{\left(<\right)}$, $\Theta_{\text{S}}^{\left(=\right)}$, and $\Theta_{\text{S}}^{\left(>\right)}$ are all non-empty. The additional condition $\Theta_{\text{F,S}}^{\left(=,=\right)}=\emptyset$ means that the two marginal-balance isoquants do not touch or cross. Under positive net interaction, this separated configuration has only two possible orderings: $\Theta_{\text{F}}^{\left(=\right)}\subseteq\Theta_{\text{S}}^{\left(<\right)}$ and equivalently $\Theta_{\text{S}}^{\left(=\right)}\subseteq\Theta_{\text{F}}^{\left(>\right)}$, or $\Theta_{\text{F}}^{\left(=\right)}\subseteq\Theta_{\text{S}}^{\left(>\right)}$ and equivalently $\Theta_{\text{S}}^{\left(=\right)}\subseteq\Theta_{\text{F}}^{\left(<\right)}$. These two orderings are exactly the cases covered by Proposition \ref{prop:tvar_second_result}. Economically, both self-protection and self-insurance have meaningful interior marginal-balance loci, but these two loci are not jointly compatible. Thus there is no feasible strategy at which the marginal TVaR benefit and marginal cost of both activities are simultaneously balanced. The optimizer is therefore selected by one of two extreme constrained marginal-balance candidates. If the frequency marginal-balance isoquant lies below the severity marginal-balance isoquant, the upper-left candidate ${\bm\theta}^{\text{UL}}$ in \eqref{eq:theta_UL} is optimal; if it lies above, the lower-right candidate ${\bm\theta}^{\text{LR}}$ in \eqref{eq:theta_LR} is optimal. In economic terms, the RH chooses the extreme combination that best reconciles the incompatible marginal incentives.

The next case occurs when $\Theta_{\text{F,S}}^{\left(=,=\right)}\neq\emptyset$, so that the two downward-sloping isoquants touch or cross inside the left-region rectangle $\Theta^{\left(\text{L}\right)}$. We now study this touching-or-crossing case in detail. Suppose that $\Theta_{\text{F,S}}^{\left(=,=\right)}=\mathcal{I}_{\text{F}}\cap\mathcal{I}_{\text{S}}\neq\emptyset$. Then there exists at least one feasible strategy $\left(p,l\right)\in\Theta^{\left(\text{L}\right)}$ at which both marginal-balance conditions hold: $G_p\left(p,l\right)=0$ and $G_l\left(p,l\right)=0$. Geometrically, such points are precisely the intersection points, or common components, of the frequency and severity marginal-balance isoquants.

When the two isoquants intersect, their common set may take different forms. They may cross at isolated points, coincide over non-degenerate intervals, or touch without changing their relative vertical order. These distinctions matter for the global solution of the left-region problem. A crossing changes the relative ordering of the two marginal-balance curves and therefore changes the local descent geometry of the objective. A touching component, by contrast, produces a common marginal-balance component without reversing the relative ordering of the two isoquants.

To describe these possibilities, we work with graph representations of the two isoquants. Whenever the isoquants are non-empty and can be represented as graphs, write
\begin{equation*}
\mathcal{I}_{\text{F}}=\left\{\left(p,g_{\text{F}}\left(p\right)\right)\in\Theta^{\left(\text{L}\right)}:p\in\mathcal{P}_{\text{F}}\right\},
\end{equation*}
and
\begin{equation*}
\mathcal{I}_{\text{S}}=\left\{\left(p,g_{\text{S}}\left(p\right)\right)\in\Theta^{\left(\text{L}\right)}:p\in\mathcal{P}_{\text{S}}\right\},
\end{equation*}
where $\mathcal{P}_{\text{F}}$ and $\mathcal{P}_{\text{S}}$ are subintervals of $\left[\underline{p},\overline{p}_{\alpha}\right]$. The functions $g_{\text{F}}$ and $g_{\text{S}}$ are implicitly defined by
\begin{equation*}
G_p\left(p,g_{\text{F}}\left(p\right)\right)=0,\quad p\in\mathcal{P}_{\text{F}},
\end{equation*}
and
\begin{equation*}
G_l\left(p,g_{\text{S}}\left(p\right)\right)=0,\quad p\in\mathcal{P}_{\text{S}}.
\end{equation*}
Under the positive net interaction assumption $G_{pl}>0$, by Lemma \ref{eq:gradient_property}, both $g_{\text{F}}$ and $g_{\text{S}}$ are strictly decreasing on their respective domains.

Let $\mathcal{P}=\mathcal{P}_{\text{F}}\cap\mathcal{P}_{\text{S}}$ be the common $p$-domain on which the two isoquants can be compared vertically. Since $\Theta_{\text{F,S}}^{\left(=,=\right)}=\mathcal{I}_{\text{F}}\cap\mathcal{I}_{\text{S}}\neq\emptyset$, we have $\mathcal{P}\neq\emptyset$. Write $\mathcal{P}=\left[\underline{p}_{\cap},\overline{p}_{\cap}\right]$. On this common domain, the relative position of the two isoquants is determined by the sign of the difference function
\begin{equation}
\Delta\left(p\right)=g_{\text{F}}\left(p\right)-g_{\text{S}}\left(p\right).
\label{eq:Delta_function}
\end{equation}
The equality set
\begin{equation}
\mathcal{P}_0=\left\{p\in\mathcal{P}:g_{\text{F}}\left(p\right)=g_{\text{S}}\left(p\right)\right\}=\left\{p\in\mathcal{P}:\Delta\left(p\right)=0\right\}
\label{eq:equality_set}
\end{equation}
is the $p$-projection of the non-empty common marginal-balance set $\Theta_{\text{F,S}}^{\left(=,=\right)}=\mathcal{I}_{\text{F}}\cap\mathcal{I}_{\text{S}}$. More precisely,
\begin{equation*}
\Theta_{\text{F,S}}^{\left(=,=\right)}=\mathcal{I}_{\text{F}}\cap\mathcal{I}_{\text{S}}=\left\{\left(p,g_{\text{F}}\left(p\right)\right)\in\Theta^{\left(\text{L}\right)}:p\in\mathcal{P}_0\right\}=\left\{\left(p,g_{\text{S}}\left(p\right)\right)\in\Theta^{\left(\text{L}\right)}:p\in\mathcal{P}_0\right\}.
\end{equation*}
The following definition formalizes what it means for the two isoquants to cross or touch on components of $\mathcal{P}_0$.

\begin{definition}\label{def:crossing_touching}
Let $\left[a,b\right]\subseteq\mathcal{P}_0$ be a non-empty connected component of the equality set $\mathcal{P}_0$. Thus, $\Delta\left(p\right)=0$, or equivalently, $g_{\text{F}}\left(p\right)=g_{\text{S}}\left(p\right)$, for all $p\in\left[a,b\right]$.
\begin{enumerate}
\item[(i)] The two isoquants $\mathcal{I}_{\text{F}}$ and $\mathcal{I}_{\text{S}}$ cross each other on the component $\left[a,b\right]$ if $\underline{p}_{\cap}<a\leq b<\overline{p}_{\cap}$, and there exists $\varepsilon>0$ such that, for any $p_L\in\left(a-\varepsilon,a\right)\cap\left[\underline{p}_{\cap},a\right)$ and any $p_R\in\left(b,b+\varepsilon\right)\cap\left(b,\overline{p}_{\cap}\right]$, we have $\Delta\left(p_L\right)\Delta\left(p_R\right)<0$. That is, the relative vertical order of the two isoquants changes when passing from the left side of $\left[a,b\right]$ to the right side.
\item[(ii)] The two isoquants $\mathcal{I}_{\text{F}}$ and  $\mathcal{I}_{\text{S}}$ touch each other on the component $\left[a,b\right]$ if either (I) $a=\underline{p}_{\cap}$, or (II) $b=\overline{p}_{\cap}$, or (III) $\underline{p}_{\cap}<a\leq b<\overline{p}_{\cap}$ and there exists $\varepsilon>0$ such that, for any $p_L\in\left(a-\varepsilon,a\right)\cap\left[\underline{p}_{\cap},a\right)$ and any $p_R\in\left(b,b+\varepsilon\right)\cap\left(b,\overline{p}_{\cap}\right]$, we have $\Delta\left(p_L\right)\Delta\left(p_R\right)>0$. That is, the two isoquants share the component $\left[a,b\right]$ without reversing their relative vertical order; a touching component may therefore attach to the left endpoint of $\mathcal{P}$, attach to the right endpoint of $\mathcal{P}$, or appear in the interior of $\mathcal{P}$.
\end{enumerate}
\end{definition}

\begin{lemma}\label{cross_touch_lemma}
Assume that $\Theta^{\left(\text{L}\right)}\neq\emptyset$, $G_{pl}>0$ on $\Theta^{\left(\text{L}\right)}$, and $\Theta_{\text{F,S}}^{\left(=,=\right)}=\mathcal{I}_{\text{F}}\cap\mathcal{I}_{\text{S}}\neq\emptyset$. Suppose that the equality set $\mathcal{P}_0$ in \eqref{eq:equality_set} has finitely many connected components. Then there exist $m,n_1,n_2,\dots,n_{m+1}\in\mathbb{N}_0=\mathbb{N}\cup\left\{0\right\}$, with at least one of them belonging to $\mathbb{N}$, such that the connected components of $\mathcal{P}_0$ can be decomposed into crossing components and touching components as follows.

If $m\in\mathbb{N}$, there are crossing intervals $\left[p^{\left(\text{C}\right)}_{i,1},p^{\left(\text{C}\right)}_{i,2}\right]\subseteq\mathcal{P}_0$, $i=1,2,\dots,m$. For each $i=1,2,\dots,m+1$, if $n_i\in\mathbb{N}$, there are touching intervals $\left[p^{\left(\text{T}\right)}_{i,j,1},p^{\left(\text{T}\right)}_{i,j,2}\right]\subseteq\mathcal{P}_0$, $j=1,2,\dots,n_i$. These components are ordered from left to right along the $p$-axis in the following way:
\begin{equation*}
\mathcal{T}_1,\mathcal{C}_1,\mathcal{T}_2,\mathcal{C}_2,\dots,\mathcal{T}_m,\mathcal{C}_m,\mathcal{T}_{m+1},
\end{equation*}
where $\mathcal{C}_i=\left[p^{\left(\text{C}\right)}_{i,1},p^{\left(\text{C}\right)}_{i,2}\right]$, $i=1,2,\dots,m$, and $\mathcal{T}_i=\cup_{j=1}^{n_i}\left[p^{\left(\text{T}\right)}_{i,j,1},p^{\left(\text{T}\right)}_{i,j,2}\right]$, $i=1,2,\dots,m+1$. Within each non-empty touching block $\mathcal{T}_i$, the touching intervals are ordered as 
\begin{equation*}
p^{\left(\text{T}\right)}_{i,1,1}\leq p^{\left(\text{T}\right)}_{i,1,2}<p^{\left(\text{T}\right)}_{i,2,1}\leq p^{\left(\text{T}\right)}_{i,2,2}<\cdots<p^{\left(\text{T}\right)}_{i,n_i,1}\leq p^{\left(\text{T}\right)}_{i,n_i,2}.
\end{equation*}
The first touching block $\mathcal{T}_1$, if non-empty, lies before the first crossing interval; the last touching block $\mathcal{T}_{m+1}$, if non-empty, lies after the last crossing interval. Empty touching blocks and crossing blocks are omitted whenever the corresponding indices do not exist.

To summarize, the crossing and touching components satisfy the following properties.
\begin{enumerate}
\item[(i)] The two isoquants $\mathcal{I}_{\text{F}}$ and  $\mathcal{I}_{\text{S}}$ cross each other on each crossing interval $\left[p^{\left(\text{C}\right)}_{i,1},p^{\left(\text{C}\right)}_{i,2}\right]$, $i=1,2,\dots,m$.
\item[(ii)] The two isoquants $\mathcal{I}_{\text{F}}$ and  $\mathcal{I}_{\text{S}}$ touch each other on each touching interval $\left[p^{\left(\text{T}\right)}_{i,j,1},p^{\left(\text{T}\right)}_{i,j,2}\right]$, $i=1,2,\dots,m+1$, $j=1,2,\dots,n_i$.
\item[(iii)] On every connected component of $\mathcal{P}\setminus\mathcal{P}_0$, the sign of the difference function $\Delta$ in \eqref{eq:Delta_function} is constant. Equivalently, on each such interval, either $g_{\text{F}}>g_{\text{S}}$ or $g_{\text{F}}<g_{\text{S}}$.
\item[(iv)] The sign of $\Delta$ in \eqref{eq:Delta_function} reverses when passing through a crossing interval and does not reverse when passing through a touching interval. Hence the relative vertical order of the two marginal-balance isoquants changes across crossing components and is preserved across touching components.
\end{enumerate}
\end{lemma}

\begin{proof}
Since the functions $g_{\text{F}}$ and $g_{\text{S}}$ are continuous on their respective domains, the difference function $\Delta=g_{\text{F}}-g_{\text{S}}$ is continuous on the common domain $\mathcal{P}$. Therefore, the equality set $\mathcal{P}_0$ is closed in $\mathcal{P}$. By assumption, $\mathcal{P}_0$ has finitely many connected components. Since $\mathcal{P}_0\neq\emptyset$, at least one such component exists, and each component is a closed interval, possibly degenerate.

By Definition \ref{def:crossing_touching}, each connected component of $\mathcal{P}_0$ is either a crossing component or a touching component. Ordering these components from left to right along the $p$-axis gives the crossing intervals $\left[p^{\left(\text{C}\right)}_{i,1},p^{\left(\text{C}\right)}_{i,2}\right]$, $i=1,2,\dots,m$, and the touching intervals $\left[p^{\left(\text{T}\right)}_{i,j,1},p^{\left(\text{T}\right)}_{i,j,2}\right]$, $i=1,2,\dots,m+1$, $j=1,2,\dots,n_i$. This proves the stated ordering.

On each connected component of $\mathcal{P}\setminus\mathcal{P}_0$, the function $\Delta$ is continuous and never equal to zero. Hence $\Delta$ has a constant sign on each such component. This proves (iii).

Finally, by Definition \ref{def:crossing_touching}, a crossing component is precisely a common component across which the sign of $\Delta$ changes, whereas a touching component is a common component across which the sign of $\Delta$ does not change, or one that attaches to an endpoint of the common domain $\mathcal{P}$. This proves (i), (ii), and (iv).
\end{proof}

For later use, we translate intervals in the equality set $\mathcal{P}_0$ in \eqref{eq:equality_set} into their corresponding common marginal-balance components in the feasible set $\Theta^{\left(\text{L}\right)}$. For any set $A\subseteq\mathcal{P}_0$, define its lifted common-isoquant component by
\begin{equation*}
\mathcal{I}\left(A\right)=\left\{\left(p,g_{\text{F}}\left(p\right)\right)\in\mathcal{I}_{\text{F}}\cap\mathcal{I}_{\text{S}}:p\in A\right\}=\left\{\left(p,g_{\text{S}}\left(p\right)\right)\in\mathcal{I}_{\text{F}}\cap\mathcal{I}_{\text{S}}:p\in A\right\}.
\end{equation*}
If $A=\emptyset$, we use the convention that $\mathcal{I}\left(A\right)=\emptyset$. For any $\left(p,l\right)\in\mathcal{I}\left(A\right)$, both first-order conditions $G_p\left(p,l\right)=0$ and $G_l\left(p,l\right)=0$ are satisfied. Therefore, if $A$ is a connected interval, such as a crossing interval or a touching interval, the gradient of $G$ is zero at every strategy of $\mathcal{I}\left(A\right)$, and hence $G$ is constant on $\mathcal{I}\left(A\right)$. This observation will be used in Propositions \ref{prop:tvar_third_result} and \ref{prop:tvar_crossing} below.

With this notation, the crossing and touching parts of the common marginal-balance set are
\begin{equation*}
\mathcal{I}^{\left(\text{C}\right)}=\cup_{i=1}^{m}\mathcal{I}\left(\left[p^{\left(\text{C}\right)}_{i,1},p^{\left(\text{C}\right)}_{i,2}\right]\right),
\end{equation*}
and
\begin{equation*}
\mathcal{I}^{\left(\text{T}\right)}=\cup_{i=1}^{m+1}\cup_{j=1}^{n_i}\mathcal{I}\left(\left[p^{\left(\text{T}\right)}_{i,j,1},p^{\left(\text{T}\right)}_{i,j,2}\right]\right),
\end{equation*}
where unions over empty index sets are interpreted as empty. Hence
\begin{equation*}
\mathcal{I}\left(\mathcal{P}_0\right)=\mathcal{I}_{\text{F}}\cap\mathcal{I}_{\text{S}}=\Theta_{\text{F,S}}^{\left(=,=\right)}=\mathcal{I}^{\left(\text{C}\right)}\cup\mathcal{I}^{\left(\text{T}\right)}.
\end{equation*}

The ordering of the intervals in Lemma \ref{cross_touch_lemma} induces the same left-to-right ordering of the lifted components in $\Theta^{\left(\text{L}\right)}$. Crossing components are precisely those common marginal-balance components across which the sign of the difference function $\Delta$ in \eqref{eq:Delta_function} reverses. Touching components are those common marginal-balance components across which the sign of $\Delta$ in \eqref{eq:Delta_function} does not reverse, or those that attach to one of the endpoints of the common comparison domain $\mathcal{P}$.

For later reference, define the upper-left endpoint candidate set by
\begin{equation}
\Theta^{\text{UL}}
=
\begin{cases}
\mathcal{I}\left(\left[p^{\left(\text{T}\right)}_{1,1,1},p^{\left(\text{T}\right)}_{1,1,2}\right]\right),&\text{if }n_1\in\mathbb{N}\text{ and }\underline{p}_{\cap}=p^{\left(\text{T}\right)}_{1,1,1},\\
\left\{{\bm\theta}^{\text{UL}}\right\},&\text{otherwise},
\end{cases}
\label{eq:ThetaUL}
\end{equation}
and define the lower-right endpoint candidate set by
\begin{equation}
\Theta^{\text{LR}}
=
\begin{cases}
\mathcal{I}\left(\left[p^{\left(\text{T}\right)}_{m+1,n_{m+1},1},p^{\left(\text{T}\right)}_{m+1,n_{m+1},2}\right]\right),
&
\text{if }n_{m+1}\in\mathbb{N}\text{ and }p^{\left(\text{T}\right)}_{m+1,n_{m+1},2}=\overline{p}_{\cap},\\
\left\{{\bm\theta}^{\text{LR}}\right\},&\text{otherwise}.
\end{cases}
\label{eq:ThetaLR}
\end{equation}
When an endpoint touching component attaches to the corresponding end of $\mathcal{P}$, the endpoint candidate is that entire touching component; otherwise, it is the singleton constrained marginal-balance candidate defined in \eqref{eq:theta_UL} and \eqref{eq:theta_LR}.

We first consider the case in which the two isoquants touch but do not cross. By Lemma \ref{cross_touch_lemma}, this corresponds to $m=0$, $n_1\in\mathbb{N}$, and all connected components of $\mathcal{P}_0$ are touching components $\left[p^{\left(\text{T}\right)}_{1,j,1},p^{\left(\text{T}\right)}_{1,j,2}\right]$, $j=1,2,\dots,n_1$.

\begin{proposition}\label{prop:tvar_third_result}
Assume that $\Theta^{\left(\text{L}\right)}\neq\emptyset$, $G_{pl}>0$ on $\Theta^{\left(\text{L}\right)}$, and $\Theta_{\text{F,S}}^{\left(=,=\right)}=\mathcal{I}_{\text{F}}\cap\mathcal{I}_{\text{S}}\neq\emptyset$. Suppose that the equality set $\mathcal{P}_0$ in \eqref{eq:equality_set} has finitely many connected components. Suppose that $m=0$ in Lemma \ref{cross_touch_lemma}. Then the following statements hold.

\begin{enumerate}
\item[(i)] If $\mathcal{P}\setminus\mathcal{T}_1\neq\emptyset$ and $\Delta\left(p\right)<0$ for all $p\in\mathcal{P}\setminus\mathcal{T}_1$, or equivalently, if $\mathcal{I}_{\text{F}}$ lies below $\mathcal{I}_{\text{S}}$ on every non-touching subinterval of $\mathcal{P}$, then the set of optimal solutions of the left-region problem \eqref{eq:TVaR_problem<=} is
\begin{equation}
\argmin_{\left(p,l\right)\in\Theta^{\left(\text{L}\right)}}G\left(p,l\right)=\Theta^{\text{UL}}.
\label{prop:tvar_third_result_result1}
\end{equation}
\item[(ii)] If $\mathcal{P}\setminus\mathcal{T}_1\neq\emptyset$ and $\Delta\left(p\right)>0$ for all $p\in\mathcal{P}\setminus\mathcal{T}_1$, or equivalently, if $\mathcal{I}_{\text{F}}$ lies above $\mathcal{I}_{\text{S}}$ on every non-touching subinterval of $\mathcal{P}$, then the set of optimal solutions of the left-region problem \eqref{eq:TVaR_problem<=} is
\begin{equation}
\argmin_{\left(p,l\right)\in\Theta^{\left(\text{L}\right)}}G\left(p,l\right)=\Theta^{\text{LR}}.
\label{prop:tvar_third_result_result2}
\end{equation}
\item[(iii)] If $\Delta\left(p\right)=0$ for all $p\in\mathcal{P}$, or equivalently, if $\mathcal{P}_0=\mathcal{P}$, then the set of optimal solutions of the left-region problem \eqref{eq:TVaR_problem<=} is
\begin{equation}
\argmin_{\left(p,l\right)\in\Theta^{\left(\text{L}\right)}}G\left(p,l\right)=\mathcal{I}\left(\mathcal{P}_0\right)=\mathcal{I}_{\text{F}}\cap\mathcal{I}_{\text{S}}=\Theta_{\text{F,S}}^{\left(=,=\right)}.
\label{prop:tvar_third_result_result3}
\end{equation}
\end{enumerate}
\end{proposition}

\begin{proof}
We prove (i). Suppose that $\mathcal{P}\setminus\mathcal{T}_1\neq\emptyset$ and $\Delta\left(p\right)<0$ for all $p\in\mathcal{P}\setminus\mathcal{T}_1$. Under this configuration, the feasible rectangle is decomposed as
\begin{equation*}
\Theta^{\left(\text{L}\right)}=\Theta_{\text{F,S}}^{\left(<,<\right)}\cup\Theta_{\text{F,S}}^{\left(=,<\right)}\cup\Theta_{\text{F,S}}^{\left(=,=\right)}\cup\Theta_{\text{F,S}}^{\left(>,<\right)}\cup\Theta_{\text{F,S}}^{\left(>,=\right)}\cup\Theta_{\text{F,S}}^{\left(>,>\right)}.
\end{equation*}
Here $\Theta_{\text{F,S}}^{\left(=,=\right)}=\mathcal{I}_{\text{F}}\cap\mathcal{I}_{\text{S}}=\mathcal{I}\left(\mathcal{P}_0\right)$ is non-empty and consists only of touching components, while the joint sign regions $\Theta_{\text{F,S}}^{\left(<,=\right)}$, $\Theta_{\text{F,S}}^{\left(<,>\right)}$, $\Theta_{\text{F,S}}^{\left(=,>\right)}$ are empty. By the same separation-geometry argument used in the proof of Proposition \ref{prop:tvar_second_result}, no strategy in $\Theta_{\text{F,S}}^{\left(<,<\right)}\cup\Theta_{\text{F,S}}^{\left(>,>\right)}$ can be optimal. Thus any optimal solution must lie in $\Theta_{\text{F,S}}^{\left(=,<\right)}\cup\Theta_{\text{F,S}}^{\left(=,=\right)}\cup\Theta_{\text{F,S}}^{\left(>,<\right)}\cup\Theta_{\text{F,S}}^{\left(>,=\right)}$.

Consider first the case $\underline{p}_{\cap}<p^{\left(\text{T}\right)}_{1,1,1}$. Then the first touching component does not attach to the left endpoint of the common comparison domain $\mathcal{P}$. The same separation-geometry argument, together with the gradient theorem, as in Proposition \ref{prop:tvar_second_result}, shows that every strategy in $\left(\Theta_{\text{F,S}}^{\left(=,<\right)}\cup\Theta_{\text{F,S}}^{\left(=,=\right)}\cup\Theta_{\text{F,S}}^{\left(>,<\right)}\cup\Theta_{\text{F,S}}^{\left(>,=\right)}\right)\setminus\left\{{\bm\theta}^{\text{UL}}\right\}$ has strictly larger objective value than ${\bm\theta}^{\text{UL}}\in\left(\Theta_{\text{F,S}}^{\left(=,<\right)}\cup\Theta_{\text{F,S}}^{\left(>,<\right)}\cup\Theta_{\text{F,S}}^{\left(>,=\right)}\right)$. Hence \eqref{prop:tvar_third_result_result1} holds with $\Theta^{\text{UL}}=\left\{{\bm\theta}^{\text{UL}}\right\}$.

Next consider the case $\underline{p}_{\cap}=p^{\left(\text{T}\right)}_{1,1,1}$. Then the first touching component attaches to the left endpoint of $\mathcal{P}$. The objective $G$ is constant on $\mathcal{I}\left(\left[p^{\left(\text{T}\right)}_{1,1,1},p^{\left(\text{T}\right)}_{1,1,2}\right]\right)$. Again, the same separation-geometry argument, together with the gradient theorem, shows that every strategy in $\left(\Theta_{\text{F,S}}^{\left(=,<\right)}\cup\Theta_{\text{F,S}}^{\left(=,=\right)}\cup\Theta_{\text{F,S}}^{\left(>,<\right)}\cup\Theta_{\text{F,S}}^{\left(>,=\right)}\right)\setminus\mathcal{I}\left(\left[p^{\left(\text{T}\right)}_{1,1,1},p^{\left(\text{T}\right)}_{1,1,2}\right]\right)$ has objective value strictly larger than the common value attained on this component $\mathcal{I}\left(\left[p^{\left(\text{T}\right)}_{1,1,1},p^{\left(\text{T}\right)}_{1,1,2}\right]\right)\subseteq\Theta_{\text{F,S}}^{\left(=,=\right)}$. Hence \eqref{prop:tvar_third_result_result1} holds with $\Theta^{\text{UL}}=\mathcal{I}\left(\left[p^{\left(\text{T}\right)}_{1,1,1},p^{\left(\text{T}\right)}_{1,1,2}\right]\right)$.
This proves (i).

The proof of (ii) is analogous. Suppose that $\mathcal{P}\setminus\mathcal{T}_1\neq\emptyset$ and $\Delta\left(p\right)>0$ for all $p\in\mathcal{P}\setminus\mathcal{T}_1$. Under this configuration, any optimal solution must lie in $\Theta_{\text{F,S}}^{\left(<,=\right)}\cup\Theta_{\text{F,S}}^{\left(<,>\right)}\cup\Theta_{\text{F,S}}^{\left(=,=\right)}\cup\Theta_{\text{F,S}}^{\left(=,>\right)}$. If the last touching component does not attach to the right endpoint of $\mathcal{P}$, the separation-geometry argument, together with the gradient theorem, selects the lower-right constrained candidate ${\bm\theta}^{\text{LR}}$. If the last touching component attaches to the right endpoint of $\mathcal{P}$, the entire last touching component $\mathcal{I}\left(\left[p^{\left(\text{T}\right)}_{1,n_1,1},p^{\left(\text{T}\right)}_{1,n_1,2}\right]\right)$ is optimal. Therefore, \eqref{prop:tvar_third_result_result2} holds. This proves (ii).

Finally, suppose that $\Delta\left(p\right)=0$ for all $p\in\mathcal{P}$. Then $\mathcal{P}_0=\mathcal{P}$. Under this configuration, the feasible rectangle is decomposed as
\begin{equation*}
\Theta^{\left(\text{L}\right)}=\Theta_{\text{F,S}}^{\left(<,<\right)}\cup\Theta_{\text{F,S}}^{\left(=,=\right)}\cup\Theta_{\text{F,S}}^{\left(>,>\right)}.
\end{equation*}
Again, no strategy in $\Theta_{\text{F,S}}^{\left(<,<\right)}\cup\Theta_{\text{F,S}}^{\left(>,>\right)}$ can be optimal. Thus any optimal solution must lie in $\Theta_{\text{F,S}}^{\left(=,=\right)}=\mathcal{I}_{\text{F}}\cap\mathcal{I}_{\text{S}}=\mathcal{I}\left(\mathcal{P}_0\right)$. The objective $G$ is constant on $\mathcal{I}\left(\mathcal{P}_0\right)$, and thus \eqref{prop:tvar_third_result_result3} holds. This proves (iii).
\end{proof}

Proposition \ref{prop:tvar_third_result} treats the case in which the two isoquants touch but do not cross. Economically, the two risk-reduction margins become jointly balanced on a common component, but the relative ordering of the marginal-balance curves does not reverse. A touching component therefore behaves like an endpoint extension of the separated-isoquant case. If the touching component attaches to the upper-left side of the relevant geometry, the upper-left strategy is replaced by an entire upper-left touching component. Similarly, if it attaches to the lower-right side, the lower-right strategy is replaced by an entire lower-right touching component. If the two isoquants coincide over the whole common domain, then every strategy on the common marginal-balance set is optimal, reflecting complete local agreement between the self-protection and self-insurance marginal conditions.

We next consider the case in which the two isoquants cross at least once. By Lemma \ref{cross_touch_lemma}, this corresponds to $m\in\mathbb{N}$, and $n_1,n_2,\dots,n_{m+1}\in\mathbb{N}_0$. In this case, for $i=1,2,\dots,m+1$, when $n_i\in\mathbb{N}$, touching components may still occur between crossing components, or attach to the two ends of the common comparison domain $\mathcal{P}$; when $n_i=0$, the corresponding touching block is empty.

As in the touching-without-crossing case, endpoint touching components, as defined in \eqref{eq:ThetaUL} and \eqref{eq:ThetaLR} may replace the endpoint candidates ${\bm\theta}^{\text{UL}}$ and ${\bm\theta}^{\text{LR}}$, as defined in \eqref{eq:theta_UL} and \eqref{eq:theta_LR}. Define the first non-touching comparison region before the first crossing component by
\begin{equation*}
\mathcal{P}_1=\left[\underline{p}_{\cap},p^{\left(\text{C}\right)}_{1,1}\right)\setminus\mathcal{T}_1.
\end{equation*}
By Lemma \ref{cross_touch_lemma}, the difference function $\Delta$ in \eqref{eq:Delta_function} has a constant sign on $\mathcal{P}_1$. This sign is crucial for determining which crossing components, together with which endpoint candidate sets, can contain global minimizers of the left-region problem \eqref{eq:TVaR_problem<=}.

\begin{proposition}\label{prop:tvar_crossing}
Assume that $\Theta^{\left(\text{L}\right)}\neq\emptyset$, $G_{pl}>0$ on $\Theta^{\left(\text{L}\right)}$, and $\Theta_{\text{F,S}}^{\left(=,=\right)}=\mathcal{I}_{\text{F}}\cap\mathcal{I}_{\text{S}}\neq\emptyset$. Suppose that the equality set $\mathcal{P}_0$ in \eqref{eq:equality_set} has finitely many connected components. Suppose that $m\in\mathbb{N}$ in Lemma \ref{cross_touch_lemma}. Then the following statements hold.

\begin{enumerate}
\item[(i)] Suppose that $m=2s+1$ for some $s\in\mathbb{N}_0$.
\begin{enumerate}
\item[(a)] If $\Delta\left(p\right)<0$ for all $p\in\mathcal{P}_1$, or equivalently, if $\mathcal{I}_{\text{F}}$ initially lies below $\mathcal{I}_{\text{S}}$ before the first crossing component except on the touching components, then the set of optimal solutions of the left-region problem \eqref{eq:TVaR_problem<=} is
\begin{equation*}
\argmin_{\left(p,l\right)\in\Theta^{\left(\text{L}\right)}}G\left(p,l\right)=\argmin_{{\bm\theta}\in\mathcal{K}_{\text{odd}}^{\left(-\right)}}G\left({\bm\theta}\right),
\end{equation*}
where the odd-number-crossing negative-initial-difference candidate set is defined by
\begin{equation*}
\mathcal{K}_{\text{odd}}^{\left(-\right)}=\Theta^{\text{UL}}\cup\left(\cup_{u=1}^{s}\mathcal{I}\left(\left[p^{\left(\text{C}\right)}_{2u,1},p^{\left(\text{C}\right)}_{2u,2}\right]\right)\right)\cup\Theta^{\text{LR}}.
\end{equation*}

\item[(b)] If $\Delta\left(p\right)>0$ for all $p\in\mathcal{P}_1$, or equivalently, if $\mathcal{I}_{\text{F}}$ initially lies above $\mathcal{I}_{\text{S}}$ before the first crossing component except on the touching components, then the set of optimal solutions of the left-region problem \eqref{eq:TVaR_problem<=} is
\begin{equation*}
\argmin_{\left(p,l\right)\in\Theta^{\left(\text{L}\right)}}G\left(p,l\right)=\argmin_{{\bm\theta}\in\mathcal{K}_{\text{odd}}^{\left(+\right)}}G\left({\bm\theta}\right),
\end{equation*}
where the odd-number-crossing positive-initial-difference candidate set is defined by
\begin{equation*}
\mathcal{K}_{\text{odd}}^{\left(+\right)}=\cup_{u=0}^{s}\mathcal{I}\left(\left[p^{\left(\text{C}\right)}_{2u+1,1},p^{\left(\text{C}\right)}_{2u+1,2}\right]\right).
\end{equation*}
\end{enumerate}

\item[(ii)] Suppose that $m=2s$ for some $s\in\mathbb{N}$.

\begin{enumerate}
\item[(a)] If $\Delta\left(p\right)<0$ for all $p\in\mathcal{P}_1$, or equivalently, if $\mathcal{I}_{\text{F}}$ initially lies below $\mathcal{I}_{\text{S}}$ before the first crossing component except on the touching components, then the set of optimal solutions of the left-region problem \eqref{eq:TVaR_problem<=} is
\begin{equation*}
\argmin_{\left(p,l\right)\in\Theta^{\left(\text{L}\right)}}G\left(p,l\right)=\argmin_{{\bm\theta}\in\mathcal{K}_{\text{even}}^{\left(-\right)}}G\left({\bm\theta}\right),
\end{equation*}
where the even-number-crossing negative-initial-difference candidate set is defined by
\begin{equation*}
\mathcal{K}_{\text{even}}^{\left(-\right)}=\Theta^{\text{UL}}\cup\left(\cup_{u=1}^{s}\mathcal{I}\left(\left[p^{\left(\text{C}\right)}_{2u,1},p^{\left(\text{C}\right)}_{2u,2}\right]\right)\right).
\end{equation*}

\item[(b)] If $\Delta\left(p\right)>0$ for all $p\in\mathcal{P}_1$, or equivalently, if $\mathcal{I}_{\text{F}}$ initially lies above $\mathcal{I}_{\text{S}}$ before the first crossing component except on the touching components, then the set of optimal solutions of the left-region problem \eqref{eq:TVaR_problem<=} is
\begin{equation*}
\argmin_{\left(p,l\right)\in\Theta^{\left(\text{L}\right)}}G\left(p,l\right)=\argmin_{{\bm\theta}\in\mathcal{K}_{\text{even}}^{\left(+\right)}}G\left({\bm\theta}\right),
\end{equation*}
where the even-number-crossing positive-initial-difference candidate set is defined by
\begin{equation*}
\mathcal{K}_{\text{even}}^{\left(+\right)}=\left(\cup_{u=0}^{s-1}\mathcal{I}\left(\left[p^{\left(\text{C}\right)}_{2u+1,1},p^{\left(\text{C}\right)}_{2u+1,2}\right]\right)\right)\cup\Theta^{\text{LR}}.
\end{equation*}
\end{enumerate}
\end{enumerate}
\end{proposition}

\begin{proof}
The proof follows from Lemma \ref{cross_touch_lemma}, Lemma \ref{eq:gradient_property}, and the same separation-geometry argument, together with the gradient theorem, used in Propositions \ref{prop:tvar_second_result} and \ref{prop:tvar_third_result}.

By Lemma \ref{cross_touch_lemma}, the sign of the difference function $\Delta$ in \eqref{eq:Delta_function} is constant on each non-touching subinterval of $\mathcal{P}$, is preserved after each touching component, and reverses after each crossing component. Thus the relative vertical order of the two marginal-balance isoquants, $\mathcal{I}_{\text{F}}$ and $\mathcal{I}_{\text{S}}$, alternates as $p$ increases across the crossing components. If $\Delta<0$ on $\mathcal{P}_1$, then $\mathcal{I}_{\text{F}}$ initially lies below $\mathcal{I}_{\text{S}}$ before the first crossing component except on the touching components. After the first crossing component, the order is reversed, and after the second crossing component, if it exists, the order is reversed back, and so on. If $\Delta>0$ on $\mathcal{P}_1$, then the same alternation occurs with the roles of the two vertical orderings reversed.

On any region where $\mathcal{I}_{\text{F}}$ lies below $\mathcal{I}_{\text{S}}$, the same argument as in Proposition \ref{prop:tvar_second_result}(i) shows that all non-candidate strategies in that region are dominated by the nearest relevant upper-left endpoint candidate or crossing component. On any region where $\mathcal{I}_{\text{F}}$ lies above $\mathcal{I}_{\text{S}}$, the same argument as in Proposition \ref{prop:tvar_second_result}(ii) shows that all non-candidate strategies in that region are dominated by the nearest relevant lower-right endpoint candidate or crossing component. In all cases, strategies in $\Theta_{\text{F,S}}^{\left(<,<\right)}\cup\Theta_{\text{F,S}}^{\left(>,>\right)}$ are ruled out by northeast or southwest improvements, respectively. Interior touching components that do not attach to an endpoint are dominated by the adjacent crossing or endpoint candidates, since touching does not reverse the sign of $\Delta$ and the objective $G$ is constant on them.

Hence it remains only to identify the candidate set to which all global minimizers must belong, by examining the sign-alteration of the difference function $\Delta$ as $p$ increases. Suppose first that $m=2s+1$ for some $s\in\mathbb{N}_0$. If $\Delta<0$ on $\mathcal{P}_1$, then the candidates are those in the upper-left endpoint candidate set, the even-indexed crossing components, and those in the lower-right endpoint candidate set. This gives $\mathcal{K}_{\text{odd}}^{\left(-\right)}$. Indeed, by the preceding dominance argument, every feasible strategy outside this set has strictly larger objective value than some strategy in this set, so (i)(a) holds. If $\Delta>0$ on $\mathcal{P}_1$, the parity is reversed, and the candidates are the odd-indexed crossing components. This gives $\mathcal{K}_{\text{odd}}^{\left(+\right)}$, and (i)(b) follows.

Now suppose that $m=2s$ for some $s\in\mathbb{N}$. If $\Delta<0$ on $\mathcal{P}_1$, the same parity argument gives the upper-left endpoint candidate set together with the even-indexed crossing components, namely $\mathcal{K}_{\text{even}}^{\left(-\right)}$. Hence (ii)(a) holds. If $\Delta>0$ on $\mathcal{P}_1$, the surviving candidates are the odd-indexed crossing components together with the lower-right endpoint candidate set, namely $\mathcal{K}_{\text{even}}^{\left(+\right)}$. Hence (ii)(b) holds.
\end{proof}

Proposition \ref{prop:tvar_crossing} treats the case in which the two marginal-balance isoquants cross. Economically, crossings represent changes in the relative strength of the frequency and severity marginal incentives. Each crossing reverses which margin is locally more urgent in the TVaR tradeoff. As a result, the relevant candidates alternate as one moves from left to right across the feasible rectangle. The initial vertical ordering of the two isoquants determines whether the first relevant candidate is an endpoint candidate or a crossing component, and the parity of the number of crossings determines whether the final relevant candidate is another endpoint candidate or a crossing component. Thus, unlike VaR, TVaR can generate genuinely joint optima involving both self-protection and self-insurance, located on crossing components where both marginal-balance conditions hold.

Combining Propositions \ref{prop:tvar_first_result}, \ref{prop:tvar_second_result}, \ref{prop:tvar_third_result}, and \ref{prop:tvar_crossing}, we obtain a complete characterization of the left-region TVaR problem \eqref{eq:TVaR_problem<=} under positive net interaction, subject to the finite-component condition on the equality set $\mathcal{P}_0$, as defined in \eqref{eq:equality_set}, in the touching-or-crossing case. These results show that, under positive net interaction, the left-region TVaR problem is driven by the relative geometry of the two marginal-balance isoquants. When one margin dominates globally, the solution lies on a boundary. When the two margins are separated, the solution is one of the two extreme constrained marginal-balance candidates. When they touch, an endpoint candidate may expand into a whole common marginal-balance component. When they cross, the optimal candidate set alternates across crossing components according to the parity of the crossings. This is the main economic difference between TVaR and VaR: TVaR creates a direct residual-frequency/residual-severity interaction through the bilinear tail-risk term $pl/\left(1-\alpha\right)$, so the optimal use of self-protection and self-insurance is governed by joint marginal-balance geometry rather than by a single probability threshold. This echoes the contrast emphasized at the end of Section \ref{sec:VaR}.

\subsubsection{Negative Net Interaction}\label{sec:negative_net}

In this subsection, assume that $G_{pl}<0$ on $\Theta^{\left(\text{L}\right)}$. Recall from Section \ref{sec:iso_geo} that, under this negative net interaction condition, the first-order isoquants $\mathcal{I}_{\text{F}}=\Theta_{\text{F}}^{\left(=\right)}$ and $\mathcal{I}_{\text{S}}=\Theta_{\text{S}}^{\left(=\right)}$ are upward-sloping whenever they are non-empty and can be represented as graphs. In this case, the first-order derivative $G_p$ increases in residual frequency $p$ but decreases in residual severity $l$, while $G_l$ decreases in $p$ but increases in $l$. Thus the separation geometry is upper-left/lower-right rather than lower-left/upper-right, in contrast to the positive-net-interaction case in Section \ref{sec:positive_net}. Motivated by this reflected separation geometry, we avoid repeating the full analysis and instead reduce the negative-net-interaction case to the already-solved positive-net-interaction case by reflecting the severity coordinate.

For ${\bm\theta}=\left(p,l\right)\in\Theta^{\left(\text{L}\right)}$, define $\lambda=\underline{l}+\overline{l}-l$, and $\hat{{\bm\theta}}=\left(p,\lambda\right)$. There is a one-to-one correspondence between ${\bm\theta}=\left(p,l\right)\in\Theta^{\left(\text{L}\right)}$ and $\hat{{\bm\theta}}=\left(p,\lambda\right)\in\hat{\Theta}^{\left(\text{L}\right)}$, where $\hat{\Theta}^{\left(\text{L}\right)}=\left[\underline{p},\overline{p}_{\alpha}\right]\times\left[\underline{l},\overline{l}\right]$. As a set, $\hat{\Theta}^{\left(\text{L}\right)}$ has the same rectangular form as $\Theta^{\left(\text{L}\right)}$, but its second coordinate is the reflected severity coordinate $\lambda$. For $\hat{{\bm\theta}}=\left(p,\lambda\right)\in\hat{\Theta}^{\left(\text{L}\right)}$, define the transformed objective
\begin{equation*}
\hat{G}\left(p,\lambda\right)=G\left(p,\underline{l}+\overline{l}-\lambda\right).
\end{equation*}
This is well-defined on $\hat{\Theta}^{\left(\text{L}\right)}$, because $\underline{l}+\overline{l}-\lambda\in\left[\underline{l},\overline{l}\right]$ whenever $\lambda\in\left[\underline{l},\overline{l}\right]$.

The first-order derivatives of $\hat{G}$ are $\hat{G}_p\left(p,\lambda\right)=G_p\left(p,\underline{l}+\overline{l}-\lambda\right)$ and $\hat{G}_{\lambda}\left(p,\lambda\right)=-G_l\left(p,\underline{l}+\overline{l}-\lambda\right)$. The second-order derivatives are $\hat{G}_{pp}\left(p,\lambda\right)=G_{pp}\left(p,\underline{l}+\overline{l}-\lambda\right)>0$, $\hat{G}_{p\lambda}\left(p,\lambda\right)=-G_{pl}\left(p,\underline{l}+\overline{l}-\lambda\right)>0$, and $\hat{G}_{\lambda\lambda}\left(p,\lambda\right)=G_{ll}\left(p,\underline{l}+\overline{l}-\lambda\right)>0$. Therefore, the transformed objective $\hat{G}$ satisfies the positive-net-interaction condition on $\hat{\Theta}^{\left(\text{L}\right)}$. The negative-net-interaction problem for $G$ is thus equivalent to a positive-net-interaction problem for $\hat{G}$.

Let $\hat{\Theta}^{*\left(\text{L}\right)}=\argmin_{\left(p,\lambda\right)\in\hat{\Theta}^{\left(\text{L}\right)}}\hat{G}\left(p,\lambda\right)$ denote the optimal set of the transformed left-region problem. By the positive-net-interaction analysis in Section \ref{sec:positive_net}, the set $\hat{\Theta}^{*\left(\text{L}\right)}$ is characterized by Propositions \ref{prop:tvar_first_result}, \ref{prop:tvar_second_result}, \ref{prop:tvar_third_result}, and \ref{prop:tvar_crossing}, applied to the transformed objective $\hat{G}$. The optimal set of the original negative-net-interaction left-region problem is obtained by transforming back. We summarize this in the following proposition; its proof has been outlined above, and hence shall be omitted.

\begin{proposition}\label{prop:tvar_negative_net}
Assume that $\Theta^{\left(\text{L}\right)}\neq\emptyset$ and $G_{pl}<0$ on $\Theta^{\left(\text{L}\right)}$. Then the set of optimal solutions of the left-region problem \eqref{eq:TVaR_problem<=} is
\begin{equation*}
\argmin_{\left(p,l\right)\in\Theta^{\left(\text{L}\right)}}G\left(p,l\right)=\left\{\left(p,l\right)\in\Theta^{\left(\text{L}\right)}:\left(p,\underline{l}+\overline{l}-l\right)\in\hat{\Theta}^{*\left(\text{L}\right)}\right\},
\end{equation*}
where $\hat{\Theta}^{*\left(\text{L}\right)}$ is characterized by Propositions \ref{prop:tvar_first_result}, \ref{prop:tvar_second_result}, \ref{prop:tvar_third_result}, and \ref{prop:tvar_crossing}, applied to the transformed objective $\hat{G}$.
\end{proposition}

Economically, the upward-sloping-isoquant case corresponds to a negative net interaction between residual frequency and residual severity. In this case, the cost-side interaction is negative enough to offset the TVaR tail-risk interaction, so the marginal-balance tradeoff between residual frequency and residual severity is reversed. The reflection argument shows that the upward-sloping geometry is not a fundamentally new optimization problem; rather, it is the downward-sloping geometry viewed through the reflected severity coordinate. The economic interpretation changes, however: the endpoint and crossing candidates identified in the transformed problem must be mapped back through $l=\underline{l}+\overline{l}-\lambda$, so their vertical positions are reflected in the original residual-severity scale. Hence upward-sloping isoquants describe a regime in which the interaction between self-protection and self-insurance reverses the direction of marginal-balance tradeoffs, but the optimization logic remains governed by the same separation, touching, and crossing principles.

\subsubsection{Zero Net Interaction}\label{sec:zero_net}

In this subsection, assume the degenerate case in which $G_{pl}=0$ on $\Theta^{\left(\text{L}\right)}$. This case occurs when the cost-side interaction exactly offsets the TVaR tail-risk interaction, namely when $\pi_{pl}=c_{xy}=-1/\left(1-\alpha\right)$. As discussed in Section \ref{sec:iso_geo}, the marginal effect of changing residual severity is independent of residual frequency, and the marginal effect of changing residual frequency is independent of residual severity. The two marginal-balance isoquants are coordinatewise: $\mathcal{I}_{\text{F}}=\Theta_{\text{F}}^{\left(=\right)}$ is vertical, while $\mathcal{I}_{\text{S}}=\Theta_{\text{S}}^{\left(=\right)}$ is horizontal.

Recall the marginal optima $l^*\left(p_0\right)$, for any fixed $p_0\in\left[\underline{p},\overline{p}_{\alpha}\right]$, in \eqref{eq:boundary_l^*}, and $p^*\left(l_0\right)$, for any fixed $l_0\in\left[\underline{l},\overline{l}\right]$, in \eqref{eq:boundary_p^*}. Since $G_{pl}=0$, the minimizer $p^*\left(l_0\right)$ is independent of the choice of $l_0$, and the minimizer $l^*\left(p_0\right)$ is independent of the choice of $p_0$. We therefore write $p^0=p^*\left(l_0\right)$, for any $l_0\in\left[\underline{l},\overline{l}\right]$, and $l^0=l^*\left(p_0\right)$, for any $p_0\in\left[\underline{p},\overline{p}_{\alpha}\right]$. Since $G_{pp}>0$ and $G_{ll}>0$, these coordinatewise minimizers are unique.

\begin{proposition}\label{prop:tvar_degenerate}
Assume that $\Theta^{\left(\text{L}\right)}\neq\emptyset$ and $G_{pl}=0$ on $\Theta^{\left(\text{L}\right)}$. Then the unique optimal solution of the left-region problem \eqref{eq:TVaR_problem<=} is ${\bm\theta}^{*\left(\text{L}\right)}=\left(p^0,l^0\right)$.
\end{proposition}

\begin{proof}
For any $\left(p,l\right)\in\Theta^{\left(\text{L}\right)}$, since $p^0$ minimizes $p\mapsto G\left(p,l\right)$, we have $G\left(p^0,l\right)\leq G\left(p,l\right)$. Since $l^0$ minimizes $l\mapsto G\left(p^0,l\right)$, we have $G\left(p^0,l^0\right)\leq G\left(p^0,l\right)$. Combining these inequalities gives $G\left(p^0,l^0\right)\leq G\left(p,l\right)$ for all $\left(p,l\right)\in\Theta^{\left(\text{L}\right)}$. Strict own-curvature gives uniqueness and concludes.
\end{proof}

Economically, the degenerate case is the knife-edge regime in which the technological interaction in the cost function exactly cancels the TVaR interaction between residual frequency and residual severity. The choice of self-protection and the choice of self-insurance become coordinatewise: the RH can determine the optimal residual frequency without considering residual severity, and can determine the optimal residual severity without considering residual frequency. Thus, unlike the positive- or negative-net-interaction cases, the solution is not governed by crossing or touching geometry. It is governed by two independent one-dimensional marginal balances.

\subsection{Confidence-Level Effects on Net Interaction}\label{sec:confidence_net_interaction}

The preceding subsections classify the left-region TVaR problem according to the uniform sign of the net interaction term $G_{pl}$ defined in \eqref{eq:G_pl}. This section studies how the confidence level $\alpha$ affects this sign. The key observation is that the TVaR tail-risk interaction $1/\left(1-\alpha\right)$ is increasing in $\alpha$. Thus, as the confidence level increases, the TVaR evaluation in the left region places more weight on the joint effect of residual frequency and residual severity. Since the left region is non-empty only when $\alpha\leq 1-\underline{p}$, this interpretation is understood within the range in which the left-region problem is relevant.

For a fixed cost technology, the sign of $G_{pl}$ is determined by the comparison between the cost-side interaction $c_{xy}$ and the TVaR interaction $1/\left(1-\alpha\right)$. Recall that, pointwise on the left-region feasible set, the positive-net-interaction case occurs when $c_{xy}>-1/\left(1-\alpha\right)$, the zero-net-interaction case occurs when $c_{xy}=-1/\left(1-\alpha\right)$, and the negative-net-interaction case occurs when $c_{xy}<-1/\left(1-\alpha\right)$. The corresponding analysis in Sections \ref{sec:positive_net}, \ref{sec:negative_net}, and \ref{sec:zero_net} applies when these inequalities hold uniformly on $\Theta^{\left(\text{L}\right)}$. If the inequality changes sign across \(\Theta^{(L)}\), the problem falls into the mixed-sign case left for future work. Thus, as $\alpha$ increases, the TVaR interaction shifts the problem toward the positive-net-interaction regime.

If the joint cost function is strictly supermodular, so that $c_{xy}>0$, then the technological interaction and the TVaR tail-risk interaction reinforce each other. In this case, $G_{pl}=c_{xy}+\frac{1}{1-\alpha}>0$, and therefore the positive-net-interaction analysis in Section \ref{sec:positive_net} applies for every confidence level $\alpha\in\left[0,1\right)$. Economically, when self-protection and self-insurance are technological substitutes in the cost technology, the cost-side interaction and the TVaR interaction both push the marginal-balance isoquants into the downward-sloping regime. Along these downward-sloping marginal-balance curves, the two risk-reduction activities exhibit a substitution-type tradeoff: more self-protection, meaning a lower residual probability, is locally balanced against less self-insurance, meaning a higher residual severity, and vice versa. Except for boundary cases in which the optimizer is one of the endpoint candidates, this substitution-type marginal tradeoff is reflected in the candidate strategies selected by the positive-net-interaction analysis in Section \ref{sec:positive_net}.

If the joint cost function is strictly submodular, so that $c_{xy}<0$, then the technological complementarity in the cost function offsets the TVaR tail-risk interaction. In this case, the sign of $G_{pl}$ depends on the strength of the negative cost-side interaction relative to $1/\left(1-\alpha\right)$. When the cost-side complementarity is mild, so that $c_{xy}$ is negative but not too negative, the TVaR interaction dominates and the positive-net-interaction case still applies; the marginal-balance geometry remains substitution-like. When the cost-side complementarity is sufficiently strong, so that $c_{xy}<-1/\left(1-\alpha\right)$, the net interaction becomes negative, and the negative-net-interaction analysis in Section \ref{sec:negative_net} applies. In that case, the marginal-balance isoquants are upward-sloping, so the two risk-reduction activities exhibit a complementarity-type marginal tradeoff: more self-protection is locally balanced with more self-insurance, and less self-protection with less self-insurance. The knife-edge case occurs when the two effects exactly cancel, giving the zero-net-interaction case in Section \ref{sec:zero_net}.

Economically, the confidence level $\alpha$ determines which net-interaction geometry is relevant. At lower confidence levels, a strong cost-side complementarity may dominate the TVaR interaction, producing upward-sloping marginal-balance isoquants. At higher confidence levels, the TVaR tail interaction becomes stronger and may dominate the cost-side complementarity, producing downward-sloping isoquants. Hence $\alpha$ affects not only the probability threshold $1-\alpha$, but also the qualitative interaction between self-protection and self-insurance in the TVaR problem.

To conclude the TVaR problem, the following theorem gives its complete solution. As stated at the beginning of this section, the net interaction matters only when $\alpha\leq 1-\underline{p}$.
\begin{theorem}\label{thm:tvar}
Let $\tilde{l}^*$ be defined as in Proposition \ref{lemma:VaR_problem>}. Whenever $\Theta^{\left(\text{L}\right)}\neq\emptyset$, define the left-region optimal set by $\Theta^{*\left(\text{L}\right)}=\argmin_{\left(p,l\right)\in\Theta^{\left(\text{L}\right)}}G\left(p,l\right)$, and, whenever $\Theta^{\left(\text{R}\right)}\neq\emptyset$, define the right-region optimal set by $\Theta^{*\left(\text{R}\right)}=\left\{\left(\overline{p},\tilde{l}^*\right)\right\}$. When $G_{pl}>0$, $\Theta^{*\left(\text{L}\right)}$ is characterized by Propositions \ref{prop:tvar_first_result}, \ref{prop:tvar_second_result}, \ref{prop:tvar_third_result}, and \ref{prop:tvar_crossing}. When $G_{pl}<0$, $\Theta^{*\left(\text{L}\right)}$ is characterized by Proposition \ref{prop:tvar_negative_net}. When $G_{pl}=0$, the set $\Theta^{*\left(\text{L}\right)}=\left\{\left(p^0,l^0\right)\right\}$. Then the optimal set of the TVaR problem \eqref{eq:TVaR_problem} is characterized as follows.
\begin{enumerate}
\item[(i)] If $\alpha\in\left[0,1-\overline{p}\right]$, then ${\bm\theta}^*\in\Theta^{*\left(\text{L}\right)}$.
\item[(ii)] If $\alpha\in\left(1-\overline{p},1-\underline{p}\right]$, then ${\bm\theta}^*\in\argmin_{{\bm\theta}\in\left(\Theta^{*\left(\text{L}\right)}\cup\Theta^{*\left(\text{R}\right)}\right)}G\left({\bm\theta}\right)$.
\item[(iii)] If $\alpha\in\left(1-\underline{p},1\right]$, then ${\bm\theta}^*\in\Theta^{*\left(\text{R}\right)}=\left\{\left(\overline{p},\tilde{l}^*\right)\right\}$.
\end{enumerate}
\end{theorem}

\begin{proof}
For $\alpha\in\left[0,1\right)$, the feasible set is decomposed into $\Theta^{\left(\text{L}\right)}$ in \eqref{eq:Theta^L} and $\Theta^{\left(\text{R}\right)}$ in \eqref{eq:Theta^R}. If $\alpha\in\left[0,1-\overline{p}\right]$, then $\overline{p}\leq 1-\alpha$, so $\Theta=\Theta^{\left(\text{L}\right)}$. Hence the problem coincides with the left-region problem, proving (i). If $\alpha\in\left(1-\overline{p},1-\underline{p}\right]$, then $\underline{p}\leq 1-\alpha<\overline{p}$, so both $\Theta^{\left(\text{L}\right)}$ and $\Theta^{\left(\text{R}\right)}$ are non-empty. The optimum is obtained by comparing the left-region optimal set $\Theta^{*\left(\text{L}\right)}$ with the unique right-region optimum $\Theta^{*\left(\text{R}\right)}$, proving (ii). If $\alpha\in\left(1-\underline{p},1\right]$, then $1-\alpha<\underline{p}$, so  $\Theta=\Theta^{\left(\text{R}\right)}$. Hence the problem coincides with the right-region problem. By Proposition \ref{lemma:VaR_problem>}, its unique optimizer is $\left(\overline{p},\tilde{l}^*\right)$. Finally, when $\alpha=1$, the TVaR objective is given by \eqref{eq:tvar_problem_alpha=1}, which coincides with the right-region objective. Therefore the same right-region solution applies. This proves (iii).
\end{proof}

\section{Illustrative Examples}\label{sec:illustrative}

This section illustrates the results in Sections \ref{sec:VaR} and \ref{sec:tvar}. Throughout, take $\underline{p}=0.05$, $\overline{p}=0.3$, $\underline{l}=0.2$, and $\overline{l}=1$. Thus, before risk reduction, the probability of a positive loss is $0.3$, and the normalized severity is $1$. The risk holder can reduce the residual loss probability to as low as $0.05$ and the residual severity to as low as $0.2$.

\subsection{Comparing Value-at-Risk and Tail Value-at-Risk Strategies}\label{sec:quad_cost}

In this subsection, we adapt the quadratic family of joint cost functions in \eqref{eq:cost_function}, introduced in Example \ref{eg:cost_function}. Recall that the quadratic specification is convenient because the sign of the parameter $c_{xy}=\delta$ determines whether self-protection and self-insurance are technological substitutes $\left(\delta>0\right)$, technological complements $\left(\delta<0\right)$, or technologically independent $\left(\delta=0\right)$ in the cost technology. Throughout this subsection, fix $A=20$, $B=5$, $a=2.05$, and $b=0.65$. Only the interaction parameter $\delta$ changes across the examples below, and the chosen values of $a$ and $b$ ensure positive marginal costs for both the positive and negative values of $\delta$ considered below. Since $1-\overline{p}=0.7$ and $1-\underline{p}=0.95$, these two thresholds determine which part of Theorem \ref{thm:var} or Theorem \ref{thm:tvar} is invoked as the confidence level changes.

\subsubsection{Threshold-Driven Strategies under Value-at-Risk}

We first illustrate the optimal strategies under VaR. For the numerical values in this subsection, we set $\delta=1$. This choice is used only to compute the objective values and to ensure positive marginal costs. As shown in Theorem \ref{thm:var}, the qualitative structure of the VaR solution does not depend on whether $\delta$ is positive, negative, or zero. Under VaR, the problem is governed by the probability threshold $1-\alpha$, rather than by the sign of the cost-side interaction. For this specification, the one-dimensional right-region severity optimizer in Proposition \ref{lemma:VaR_problem>} is $\tilde{l}^*=0.93$.

Table \ref{tab:var_illustration} summarizes the optimal VaR strategies for different confidence levels.
\begin{table}[htbp]
\centering
\begin{tabular}{c c c c}
\hline
\(\alpha\) & Result invoked & \({\bm\theta}^*=\left(p^*,l^*\right)\) & $F$ \\
\hline
0.5 & Theorem \ref{thm:var}(i) & \(\left(0.3,1\right)\) & 0 \\
0.7 & Theorem \ref{thm:var}(i) & \(\left(0.3,1\right)\) & 0 \\
0.85 & Theorem \ref{thm:var}(ii) & \(\left(0.15,1\right)\) & 0.5325 \\
0.9 & Theorem \ref{thm:var}(ii) & \(\left(0.1,1\right)\) & 0.81 \\
0.95 & Theorem \ref{thm:var}(ii) & \(\left(0.3,0.93\right)\) & 0.9878 \\
0.99 & Theorem \ref{thm:var}(iii) & \(\left(0.3,0.93\right)\) & 0.9878 \\
\hline
\end{tabular}
\caption{VaR optimal strategies under the quadratic cost specification.}
\label{tab:var_illustration}
\end{table}

\noindent
Table \ref{tab:var_illustration} illustrates the fully threshold-driven nature of the VaR strategies. When $\alpha\leq 1-\overline{p}=0.7$, Theorem \ref{thm:var}(i) applies: the positive loss is outside the VaR tail for all feasible residual probabilities, and thus the status quo $\left(\overline{p},\overline{l}\right)=\left(0.3,1\right)$ is optimal. When $0.7=1-\overline{p}<\alpha\leq 1-\underline{p}=0.95$, Theorem \ref{thm:var}(ii) applies: the risk holder compares a pure self-protection threshold strategy with a pure self-insurance strategy. At $\alpha=0.85$ and $\alpha=0.9$, the pure self-protection threshold strategy is optimal. At $\alpha=0.95$, the threshold strategy becomes sufficiently costly that the pure self-insurance strategy $\left(\overline p,\tilde l^*\right)=\left(0.3,0.93\right)$ is selected. When $\alpha>1-\underline{p}=0.95$, Theorem \ref{thm:var}(iii) applies: all feasible residual probabilities lie above the VaR threshold, and the same pure self-insurance strategy is optimal.

Thus, even though the cost technology allows interaction between self-protection and self-insurance, VaR selects strategies that do not require their simultaneous use. This illustrates the substitution-type structure of the VaR solution.

\subsubsection{Strategies under Tail Value-at-Risk with Cost-Side Substitutes}

We next illustrate the optimal strategies under TVaR. We keep $\delta=1>0$, so self-protection and self-insurance are technological substitutes in the cost technology. The net interaction term on the left region is $G_{pl}=1+\frac{1}{1-\alpha}>0$ for every $\alpha\in\left[0,1\right)$. Hence, whenever the TVaR left-region problem is relevant, the positive-net-interaction analysis in Section \ref{sec:positive_net} applies.

Table \ref{tab:tvar_positive_illustration} summarizes the optimal TVaR strategies for different confidence levels. For compactness, ``Thm.'' abbreviates Theorem and ``Prop.'' abbreviates Proposition.
\begin{table}[htbp]
\centering
\begin{tabular}{c c c c c}
\hline
\(\alpha\) & \(G_{pl}\) & Result(s) invoked & \({\bm\theta}^*=\left(p^*,l^*\right)\) & $G$ \\
\hline
0.5 & 3 & Thm. \ref{thm:tvar}(i), Prop. \ref{prop:tvar_first_result}(i)/(iii) & \(\left(0.3,1\right)\) & 0.6 \\
0.6 & 3.5 & Thm. \ref{thm:tvar}(i), Prop. \ref{prop:tvar_crossing}(i)(b) & \(\left(0.2783,0.9952\right)\) & 0.7449 \\
0.7 & 4.3333 & Thm. \ref{thm:tvar}(i), Prop. \ref{prop:tvar_crossing}(i)(b) & \(\left(0.2397,0.9823\right)\) & 0.9582 \\
0.8 & 6 & Thm. \ref{thm:tvar}(ii), Prop. \ref{prop:tvar_second_result}(i) & \(\left(0.3,0.93\right)\) & 0.9878 \\
0.99 & 101 & Thm. \ref{thm:tvar}(iii) & \(\left(0.3,0.93\right)\) & 0.9878 \\
\hline
\end{tabular}
\caption{TVaR optimal strategies under the quadratic cost specification with cost-side substitutes.}
\label{tab:tvar_positive_illustration}
\end{table}

\noindent
For $\alpha=0.5$, the entire feasible rectangle lies in the left region and Theorem \ref{thm:tvar}(i) applies; in this case, the status quo is optimal. For $\alpha=0.6$ and $\alpha=0.7$, the left-region problem remains relevant and $G_{pl}>0$. In these two cases, the optimal strategy uses both self-protection and self-insurance: both the residual probability and the residual severity are reduced from their baseline values. Geometrically, these optimal strategies correspond to interior crossing points of the two marginal-balance isoquants, where both marginal-balance conditions are satisfied.

For $\alpha=0.8$, both $\Theta^{\left(\text{L}\right)}$ and $\Theta^{\left(\text{R}\right)}$ are non-empty, so Theorem \ref{thm:tvar}(ii) applies. Although the left-region problem still has positive net interaction, the global TVaR optimum is obtained by comparing the left-region optimizer with the right-region optimizer. In this case, the right-region solution \(\left(\overline{p},\tilde {l}^*\right)=\left(0.3,0.93\right)\) is globally optimal. For \(\alpha=0.99\), all feasible residual probabilities lie in the right region, so Theorem \ref{thm:tvar}(iii) applies and the same right-region solution is optimal.

This example illustrates two points. First, unlike VaR, TVaR can induce genuinely joint self-protection/self-insurance strategies in the left region. Second, the global TVaR solution still depends on the left-region/right-region decomposition in Theorem \ref{thm:tvar}; an interior left-region marginal-balance strategy need not be globally optimal when the right-region candidate has a lower objective value.

Finally, when self-protection and self-insurance are technologically independent in the cost technology, the findings are similar to those in this subsection because \(G_{pl}=1/(1-\alpha)>0\) for all \(\alpha\in[0,1)\).

\subsubsection{Strategies under Tail Value-at-Risk with Cost-Side Complements}

We now set $\delta=-2.5<0$, so self-protection and self-insurance are technological complements in the cost technology. The net interaction term on the left region is $G_{pl}=-2.5+\frac{1}{1-\alpha}$, $\alpha\in\left[0,1\right)$. Hence, $G_{pl}<0$ for $\alpha\in\left[0,0.6\right)$, $G_{pl}=0$ for $\alpha=0.6$, and $G_{pl}>0$ for $\alpha\in\left(0.6,1\right)$. At lower confidence levels, the cost-side complementarity dominates the TVaR tail-risk interaction, and the marginal-balance isoquants are upward-sloping. At the knife-edge confidence level $\alpha=0.6$, the cost-side complementarity exactly offsets the TVaR tail-risk interaction, and the isoquants become coordinatewise. At higher confidence levels, the TVaR tail-risk interaction dominates, and the isoquants become downward-sloping.

Table \ref{tab:tvar_confidence_illustration} summarizes the optimal TVaR strategies for different confidence levels. For compactness, ``Thm.'' abbreviates Theorem and ``Prop.'' abbreviates Proposition. In the rows with $G_{pl}<0$, ``refl. Prop.'' refers to the proposition applied to the transformed objective $\hat{G}$ after the severity reflection $\lambda=\underline{l}+\overline{l}-l$.

\begin{table}[htbp]
\centering
\small
\begin{tabular}{c c c c c}
\hline
\(\alpha\) & \(G_{pl}\) & Result(s) invoked & \({\bm\theta}^*=\left(p^*,l^*\right)\) & \(G\) \\
\hline
0.5 & \(-0.5\) & Thm. \ref{thm:tvar}(i), Prop. \ref{prop:tvar_negative_net}, refl. Prop. \ref{prop:tvar_second_result}(ii) & \(\left(0.3,1\right)\) & 0.6 \\
0.55 & \(-0.2778\) & Thm. \ref{thm:tvar}(i), Prop. \ref{prop:tvar_negative_net}, refl. Prop. \ref{prop:tvar_crossing}(i)(b) & \(\left(0.2913,0.9962\right)\) & 0.6659 \\
0.58 & \(-0.119\) & Thm. \ref{thm:tvar}(i), Prop. \ref{prop:tvar_negative_net}, refl. Prop. \ref{prop:tvar_crossing}(i)(b) & \(\left(0.2834,0.9867\right)\) & 0.7111 \\
0.59 & \(-0.061\) & Thm. \ref{thm:tvar}(i), Prop. \ref{prop:tvar_negative_net}, refl. Prop. \ref{prop:tvar_crossing}(i)(b) & \(\left(0.2805,0.9834\right)\) & 0.7272 \\
0.6 & 0 & Thm. \ref{thm:tvar}(i), Prop. \ref{prop:tvar_degenerate} & \(\left(0.2775,0.98\right)\) & 0.7439 \\
0.7 & 0.8333 & Thm. \ref{thm:tvar}(i), Prop. \ref{prop:tvar_crossing}(i)(b) & \(\left(0.2383,0.9403\right)\) & 0.95 \\
0.8 & 2.5 & Thm. \ref{thm:tvar}(ii), Prop. \ref{prop:tvar_crossing}(i)(b) & \(\left(0.3,0.93\right)\) & 0.9878 \\
0.99 & 97.5 & Thm. \ref{thm:tvar}(iii) & \(\left(0.3,0.93\right)\) & 0.9878 \\
\hline
\end{tabular}
\caption{TVaR optimal strategies under the quadratic cost specification with cost-side complements.}
\label{tab:tvar_confidence_illustration}
\end{table}

\noindent
For $\alpha=0.5$, the entire feasible rectangle lies in the left region and the net interaction is negative. Hence Theorem \ref{thm:tvar}(i) applies together with Proposition \ref{prop:tvar_negative_net}; in this case, the status quo is optimal. For $\alpha=0.55$, $\alpha=0.58$, and $\alpha=0.59$, the entire feasible rectangle still lies in the left region and $G_{pl}<0$, so Theorem \ref{thm:tvar}(i) again applies together with Proposition \ref{prop:tvar_negative_net}. In these three cases, the optimal strategy uses both self-protection and self-insurance. Thus, the upward-sloping isoquant case can also induce genuine joint self-protection/self-insurance optima, rather than only boundary or status-quo solutions.

At $\alpha=0.6$, the entire feasible rectangle still lies in the left region, but the net interaction is zero. Hence Theorem \ref{thm:tvar}(i) applies, and the two marginal-balance isoquants become coordinatewise. Proposition \ref{prop:tvar_degenerate} therefore applies to the left-region problem. At $\alpha=0.7$, the entire feasible rectangle still lies in the left region, but the net interaction becomes positive. Thus Theorem \ref{thm:tvar}(i) applies together with the positive-net-interaction analysis in Section \ref{sec:positive_net}; in this case, the two marginal-balance isoquants are downward-sloping. In both cases, the optimal strategy again uses both self-protection and self-insurance.

For $\alpha=0.8$, both $\Theta^{\left(\text{L}\right)}$ and $\Theta^{\left(\text{R}\right)}$ are non-empty, so Theorem \ref{thm:tvar}(ii) applies. Comparing the left-region optimizer with the right-region optimizer, the globally optimal strategy is the right-region solution \(\left(\overline{p},\tilde {l}^*\right)=\left(0.3,0.93\right)\). For \(\alpha=0.99\), all feasible residual probabilities lie in the right region, so Theorem \ref{thm:tvar}(iii) applies and the same right-region solution is optimal.

This example illustrates Propositions \ref{prop:tvar_negative_net} and \ref{prop:tvar_degenerate}. When the cost-side complementarity is strong enough relative to the TVaR tail-risk interaction, the net interaction is negative and the marginal-balance isoquants are upward-sloping. In this case, the optimal strategy may still be an interior joint self-protection/self-insurance strategy. As the confidence level increases, the TVaR tail-risk interaction becomes stronger until it exactly offsets the cost-side complementarity. In this knife-edge case, the optimal strategy is determined by coordinatewise best self-protection and self-insurance choices. As the confidence level increases further, the net interaction turns positive, and the problem shifts back to the downward-sloping isoquant geometry studied in Section \ref{sec:positive_net}.

\subsection{Higher-Order Joint Cost Function}

The quadratic cost specification in Section \ref{sec:quad_cost} is useful for illustrating the comparison between VaR and TVaR solutions and for showing how TVaR strategies differ under cost-side substitutes and complements. However, in the left-region TVaR problem, a quadratic cost function generates affine marginal-balance isoquants, so the two isoquants can cross at most once. To illustrate the multiple-crossing case in Proposition \ref{prop:tvar_crossing}, we now consider a higher-order joint cost specification. Throughout this subsection, set $\alpha=0.5$. Since $\alpha=0.5<1-\overline{p}=0.7$, the entire feasible rectangle lies in the left region, so $\Theta=\Theta^{\left(\text{L}\right)}$, and Theorem \ref{thm:tvar}(i) applies.

Define, for any ${\bm\theta}\in\Theta$,
\begin{align*}
\pi\left(p,l\right)=&\;\Phi\left(p\right)+pl+l^2-\frac{5}{2}l-K\\=&\;25p^4-23p^3+10.06p^2-4.9092p+pl+l^2-\frac{5}{2}l-K,
\end{align*}
where $K$ is the normalizing constant such that $\pi\left(\overline{p},\overline{l}\right)=\pi\left(0.3,1\right)=0$. A direct check gives $\pi_p<0$, $\pi_l<0$, $\pi_{pp}>0$, and $\pi_{ll}>0$ on $\Theta$. Thus the maintained monotonicity and own-curvature assumptions are satisfied. Equivalently, this defines a higher-order joint cost function in reduction variables through \(c\left(x,y\right)=\pi\left(\overline{p}-x,\overline{l}-y\right)\).

Since $\alpha=0.5$, the left-region TVaR objective is, for any $\left(p,l\right)\in\Theta$,
\begin{equation*}
G\left(p,l\right)=\pi\left(p,l\right)+2pl=\Phi\left(p\right)+3pl+l^2-\frac{5}{2}l-K.
\end{equation*}
Therefore, on $\Theta$, $G_{pl}=3>0$, and the positive-net-interaction analysis in Section \ref{sec:positive_net} applies. The two marginal-balance isoquants are
\begin{equation*}
\mathcal{I}_{\text{F}}=\left\{\left(p,l\right)\in\Theta:G_p\left(p,l\right)=\Phi'\left(p\right)+3l=0\right\},
\end{equation*}
and
\begin{equation*}
\mathcal{I}_{\text{S}}=\left\{\left(p,l\right)\in\Theta:G_l\left(p,l\right)=3p+2l-\frac{5}{2}=0\right\}.
\end{equation*}
The two isoquants $\mathcal{I}_{\text{F}}$ and $\mathcal{I}_{\text{S}}$ cross three times inside the feasible rectangle, at approximately $\left(0.18,0.98\right)$, $\left(0.23,0.905\right)$, and $\left(0.28,0.83\right)$. Table \ref{tab:multiple_crossings_polynomial} summarizes the objective values at these three crossings and indicates which crossings are selected by the candidate-set geometry in Proposition \ref{prop:tvar_crossing}.
\begin{table}[htbp]
\centering
\begin{tabular}{c c c c c}
\hline
Crossing & \(\left(p,l\right)\) & \(G\) & Selected candidate & Global optimum\\
\hline
1 & $\left(0.18,0.98\right)$ & 0.559856 & Yes & Yes\\
2 & $\left(0.23,0.905\right)$ & 0.560012 & No & No\\
3 & $\left(0.28,0.83\right)$ & 0.559856 & Yes & Yes\\
\hline
\end{tabular}
\caption{TVaR optimal strategies under the higher-order cost specification with cost-side substitutes at $\alpha=0.5$.}
\label{tab:multiple_crossings_polynomial}
\end{table}

This example illustrates why Proposition \ref{prop:tvar_crossing} is needed. With three isolated crossing components, the common marginal-balance set cannot be summarized by a single crossing point. Instead, the relative vertical ordering of \(\mathcal I_{\mathrm F}\) and \(\mathcal I_{\mathrm S}\) changes as $p$ increases across the feasible rectangle. In the notation of Proposition \ref{prop:tvar_crossing}, the number of crossings is odd, with \(m=3=2s+1\) and \(s=1\). The relevant candidate set is therefore determined by the odd-crossing case in Proposition \ref{prop:tvar_crossing}(i)(b), according to the initial vertical ordering of the two isoquants. In the present example, the first and third crossings are selected as candidates, whereas the second crossing is not.

The example also shows that the left-region TVaR problem may have non-unique local optima and, in this case, non-unique global optima. The first and third crossings are both local minima and attain the same objective value. The second crossing satisfies both marginal-balance conditions, but it is a saddle point and has a slightly larger objective value. Thus, when cost interactions generate multiple crossings, the set of minimizers need not be a singleton. Economically, the risk holder may face two distinct optimal ways of balancing self-protection and self-insurance: one with more self-protection and less self-insurance, and another with less self-protection and more self-insurance. This reflects the substitution-type marginal tradeoff under positive net interaction. The isoquant geometry identifies both possibilities and reduces the global problem to a comparison over the structured candidate set.

Finally, we emphasize that non-unique local optima do not necessarily imply non-unique global optima. The numerical values in this example are chosen so that the first and third crossings are both local and global optima, thereby illustrating the possibility of non-unique substitution-type optimal strategies under positive net interaction. In general, the objective values of the selected local candidates must be compared to determine the global optimum. Under negative net interaction, the same logic applies after reflecting the severity coordinate. If non-unique global optima arise in the reflected problem, then, after mapping back to the original residual-risk variables, they correspond to complementarity-type optimal strategies.



\section{Concluding Remarks and Future Directions}\label{sec:conclusion}

This paper studies how a risk holder should strategically combine self-protection and self-insurance when market insurance is set aside. In the Bernoulli loss model considered here, self-protection reduces the residual probability of a positive loss, while self-insurance reduces the residual severity conditional on loss occurrence. The risk holder incurs a joint risk-reduction cost that allows technological interaction between the two activities and evaluates residual risk using either VaR or TVaR.

Under VaR, the problem is fully threshold-driven. The optimal strategy is either no risk reduction, a pure self-protection threshold strategy, or a pure self-insurance strategy. The qualitative form of the VaR solution is not affected by whether self-protection and self-insurance are technological substitutes or complements in the cost function. Thus, VaR exhibits a substitution-type structure between the two risk-reduction strategies: the optimizer does not require their simultaneous use.

The TVaR problem is fundamentally different. On the left-region feasible set, TVaR introduces a direct interaction between residual frequency and residual severity, making the objective generally non-convex even in the Bernoulli setting. To solve this problem, we develop an isoquant geometry method based on the marginal-balance curves for self-protection and self-insurance. The sign of the net interaction term, which combines the cost-side technological interaction and the TVaR tail-risk interaction, determines whether the isoquants are downward-sloping, upward-sloping, or coordinatewise. This geometry identifies boundary candidates, extreme constrained marginal-balance candidates, touching components, and crossing components. Economically, the analysis shows that self-protection and self-insurance may behave as substitutes or complements in the optimal risk-reduction strategy depending on the balance between the cost technology, the confidence level, and the TVaR tail-risk interaction.

The illustrative examples reinforce these analytical conclusions. Under a common quadratic cost specification, they show that VaR remains fully threshold-driven, whereas TVaR can generate joint self-protection/self-insurance strategies. They also show how the confidence level can move the TVaR problem from negative net interaction, to zero net interaction, and then to positive net interaction. Finally, a higher-order cost example demonstrates that the marginal-balance isoquants may cross multiple times, leading to a structured candidate set and, in some cases, non-unique optimal joint risk-reduction strategies that exhibit substitution-type or complementarity-type behavior.

Several directions are natural for future research. First, the Bernoulli loss model can be extended to non-Bernoulli risks with multiple loss levels, continuous severities, or richer frequency-severity dependence. Second, the analysis can be developed for more general risk measures beyond VaR and TVaR, such as distortion risk measures, expectiles, mean-variance criteria, or entropic risk measures. Third, market insurance can be reintroduced to study how self-protection and self-insurance interact with insurance demand, premium principles, deductibles, coinsurance, and coverage limits. Finally, dynamic and multi-period versions of the model would allow risk-reduction decisions to interact with learning, capital constraints, and future risk exposures.

\bibliographystyle{apalike}
\bibliography{ref}

\end{document}